\documentclass[aps,pra,10pt,superscriptaddress,twocolumn,nofootinbib]{revtex4}

\usepackage{color}
\usepackage{amsmath,amsfonts,amsthm}
\usepackage{stmaryrd}
\usepackage{graphicx}

\usepackage{bm}  % Define \bm{} to use bold math fonts
\usepackage{bbm} % for double-struck identity symbol

\usepackage[pdfpagelabels,pdftex,bookmarks,breaklinks]{hyperref}
\definecolor{deeppurple}{RGB}{100,0,120} % choose colors
\definecolor{darkgreen}{RGB}{0,150,0}
\definecolor{darkblue}{RGB}{0,0,130}
\hypersetup{colorlinks=true, linkcolor=deeppurple, citecolor=darkgreen, filecolor=red, urlcolor=darkblue}
\hypersetup{
	pdftitle={Quantum Tomography via Compressed Sensing: 
		Error Bounds, Sample Complexity, and Efficient Estimators},
	pdfauthor={Steven T. Flammia, David Gross, Yi-Kai Liu, Jens Eisert}
}% add a title and authors to the metadata

\newtheorem{theorem}{Theorem}
\newtheorem{lemma}[theorem]{Lemma}

\renewenvironment{proof}[1][Proof]{\noindent\textbf{#1:} }{\ $\Box$}

\theoremstyle{remark}

\newcommand*{\C}{\mathbb{C}}
\newcommand*{\E}{\mathbb{E}}
\newcommand*{\one}{\mathbbm{1}} % needs \usepackage{bbm}

\def\EE{\mathbbm{E}}

\newcommand*{\Eref}[1]{Eq.~(\ref{#1})}

\newcommand*{\eps}{\epsilon}
\newcommand*{\veps}{\varepsilon}
\newcommand*{\Tr}{\mathrm{Tr}}
\newcommand*{\tr}{\mathrm{tr}}
\newcommand*{\e}{\mathrm{e}}

\newcommand*{\ket}[1]{|{#1}\rangle}
\newcommand*{\bra}[1]{\langle{#1}|}
\newcommand*{\ketbra}[2]{|{#1}\rangle\!\langle{#2}|}

\newcommand*{\proj}[1]{\ketbra{#1}{#1}}

\providecommand{\abs}[1]{\lvert#1\rvert} 
\providecommand{\norm}[1]{\lVert#1\rVert}

\newcommand{\set}[1]{\lbrace #1 \rbrace}

\DeclareMathOperator{\rank}{rank}
\DeclareMathOperator{\poly}{poly}

\def\Id{\mathbbm{1}}

\def\P{\mathcal{P}}

\def\calA{\mathcal{A}} % Renamed to \calA to avoid conflict with \A defined above

\newcommand*{\st}{\ \ \mbox{s.t.} \ }

\newcommand{\calE}{\mathcal{E}}
\newcommand{\calI}{\mathcal{I}}
\newcommand{\calP}{\mathcal{P}}
\newcommand{\CC}{\mathbb{C}}
\newcommand{\HH}{\mathbb{H}}
\newcommand{\RR}{\mathbb{R}}

\newcommand{\ox}{\otimes}
\newcommand{\ds}{\mathrm{DS}}
\newcommand{\lasso}{\mathrm{Lasso}}

\begin{document}

\title{Quantum Tomography via Compressed Sensing:\\
Error Bounds, Sample Complexity, and Efficient Estimators}

\author{Steven T.\ Flammia}
\affiliation{Department of Computer Science and Engineering, University of Washington, Seattle, WA, USA}
\author{David Gross}
\affiliation{Institute of Physics, University of Freiburg, 79104 Freiburg, Germany}
\author{Yi-Kai Liu}
\affiliation{National Institute of Standards and Technology, Gaithersburg, MD, USA}
\author{Jens Eisert}
\affiliation{Dahlem Center for Complex Quantum Systems, Freie Universit\"{a}t Berlin, 14195 Berlin, Germany}

\date{May 18, 2012}

\begin{abstract}
Intuitively, if a density operator has small rank, then it should be easier to estimate from experimental data, since in this case only a few eigenvectors need to be learned. We prove two complementary results that confirm this intuition. First, we show that a low-rank density matrix can be estimated using fewer copies of the state, i.e., the \textit{sample complexity} of tomography decreases with the rank. Second, we show that unknown low-rank states can be reconstructed from an \textit{incomplete} set of measurements, using techniques from compressed sensing and matrix completion. These techniques use simple Pauli measurements, and their output can be certified without making any assumptions about the unknown state.

We give a new theoretical analysis of compressed tomography, based on the \textit{restricted isometry property} (RIP) for low-rank matrices. Using these tools, we obtain near-optimal error bounds, for the realistic situation where the data contains noise due to finite statistics, and the density matrix is full-rank with decaying eigenvalues. We also obtain upper-bounds on the sample complexity of compressed tomography, and almost-matching lower bounds on the sample complexity of any procedure using adaptive sequences of Pauli measurements.

Using numerical simulations, we compare the performance of two compressed sensing estimators---the matrix Dantzig selector and the matrix Lasso---with standard maximum-likelihood estimation (MLE). We find that, given comparable experimental resources, the compressed sensing estimators consistently produce higher-fidelity state reconstructions than MLE.  In addition, the use of an incomplete set of measurements leads to faster classical processing with no loss of accuracy.

Finally, we show how to certify the accuracy of a low rank estimate using direct fidelity estimation and we describe a method for compressed quantum process tomography that works for processes with small Kraus rank, and requires only Pauli eigenstate preparations and Pauli measurements.  
\end{abstract}

\maketitle
%------------------------------------------------------------------------------------------------------------%

% !TEX root = ../tradeoff.tex

%------------------------------------------------------------------------------------------------------------%
\section{Introduction}
%------------------------------------------------------------------------------------------------------------%

In recent years there has been amazing progress in studying complex quantum mechanical systems under controlled laboratory conditions~\cite{Nature2008}. Quantum tomography of states and processes in an invaluable tool used in many such experiments, because it enables a complete characterization of the state of a quantum system or process (see e.g.~\cite{Smithey1993, Altepeter2003, OBrien2003, OBrien2004, Roos2004, Resch2005, Haffner2005, Myrskog2005, Smith2006, Riebe2006, Filipp2009, Medendorp2011, Barreiro2011, Fedorov2012, Liu2012}). Unfortunately, tomography is very resource intensive, and scales exponentially with the size of the system.  For example, a system of $n$ spin-1/2 particles (qubits) has a Hilbert space with dimension $d = 2^n$, and the state of the system is described by a density matrix with $d^2 = 4^n$ entries.

Recently a new approach to tomography was proposed:
\textit{compressed} quantum tomography, based on techniques from
compressed sensing~\cite{Gross2010,Gross2011}. The basic idea is to concentrate
on states that are well approximated by density matrices of rank $r \ll
d$.  
This approach can be applied to many realistic experimental situations, 
where the ideal state of the system is pure, and physical constraints 
(e.g., low temperature, or the locality of interactions) ensure that 
the actual (noisy) state still has low entropy. 
% (and not a maximally mixed state on the full Hilbert space). 

This approach is convenient because it does not require detailed knowledge 
about the system. However, note that when such knowledge is available, 
one can use alternative formulations of compressed tomography,
with different notions of sparsity, to further reduce the dimensionality of the
problem~\cite{Shabani2011}. We will compare these methods in
Section~\ref{S:process-related-work}.  
% Although the state is mostly supported on a low-dimensional subspace, 
% one assumes no prior knowledge about the structure of this subspace.

The main challenge in compressed tomography is how to exploit this low-rank structure, 
when one does not know the subspace on which the state is supported. 
Consider the example of a pure quantum state. Since pure states are
specified by only $O(d)$ numbers, it seems plausible that one could be
reconstructed after measuring only $O(d)$ observables, compared with
$O(d^2)$ for a general mixed state.  While this intuition is indeed
correct~\cite{Amiet1999, Flammia2005, Merkel2010, Heinosaari2011}, it is a challenge to devise a
practical tomography scheme that takes advantage of this. In
particular, one is restricted to those measurements that can be easily
performed in the lab; furthermore, one then has to find a pure state
consistent with measured data~\cite{Kaznady2009}, preferably by some
procedure that is computationally efficient (note that finding
minimum-rank solutions is NP-hard in general~\cite{Natarajan1995}).

Compressed tomography provides a solution that meets all of these
practical requirements~\cite{Gross2010,Gross2011}.  It requires
measurements of Pauli observables, which are feasible in many
experimental systems.  In total, it uses a random subset of $m =
O(rd\log d)$ Pauli observables, which is
just slightly more than the $O(rd)$ degrees of freedom that specify an
arbitrary rank $r$ state.  Then the density matrix $\rho$ is
reconstructed by solving a convex program.
This can be done efficiently using general-purpose SDP
solvers~\cite{Sturm1999}, or specialized algorithms for larger
instances~\cite{Cai2010, Becker2010, Ma2011}. The scheme is robust to noise and continues to perform well when the measurements are imprecise or when the state is only close to a low-rank state.  

Here we follow up on Refs.~\cite{Gross2010,Gross2011} by giving a stronger (and completely different) theoretical analysis, which is based on the \textit{restricted isometry property} (RIP) \cite{Liu2011, Candes2011, Recht2007}.  This answers a number of questions that could not be addressed satisfactorily using the earlier techniques based on dual certificates.  First, how large is the error in our estimated density matrix, when the true state is full-rank with decaying eigenvalues? We show that the error is not much larger than the ``tail'' of the eigenvalue spectrum of the true state. Second, how large is the \textit{sample complexity} of compressed tomography, i.e., how many copies of the unknown state are needed, to estimate its density matrix?  We show that compressed tomography achieves nearly the optimal sample complexity among all procedures using Pauli measurements, and, surprisingly, the sample complexity of compressed tomography is nearly independent of the number of measurement settings $m$, so long as $m \geq \Omega(rd \poly\log d)$.

In addition, we use numerical simulations to investigate the question: given a fixed time $T$ during which an experiment can be run, is it better to do compressed tomography or full tomography, i.e., is it better to use a few measurement settings and repeat them many times, or do all possible measurements with fewer repetitions?  For the situations we simulate, we find that compressed tomography provides better accuracy at a reduced computational cost compared to standard maximum-likelihood estimation.

Finally, we provide two useful tools:  a procedure for certifying the accuracy of low-rank state estimates, and a very simple compressed sensing technique for quantum process tomography.

We now describe these results in more detail.

\textit{Theoretical analysis using RIP.} 
% First, we prove much stronger error bounds, using a different technique than before, known as the restricted isometry property (RIP) \cite{Liu2011, Candes2011, Recht2007}.  These error bounds are nearly optimal, though they require slightly more measurement settings with $m = O(rd\log^6 d)$, a bound that we believe can be improved.
% Roughly speaking, previous work used convex duality to characterize the solution to the convex optimization.  
We use a fundamental geometric fact:  the manifold of rank-$r$ matrices in $\CC^{d\times d}$ can be embedded into $O(rd\poly\log d)$ dimensions, with small distortion in the 2-norm.  An embedding that does this is said to satisfy the \textit{restricted isometry property} (RIP) \cite{Recht2007}.  In \cite{Liu2011} it was shown that such an embedding can be realized using the expectation values of a random subset of $O(rd\poly\log d)$ Pauli matrices.  This implies the existence of so-called ``universal'' methods for low-rank matrix recovery:  there exists a \textit{fixed} set of $O(rd\poly\log d)$ Pauli measurements, that has the ability to reconstruct \textit{every} rank-$r$ $d\times d$ matrix.  Moreover, with high probability, a random choice of Pauli measurements will achieve this.  (The earlier results of \cite{Gross2010} placed the quantifiers in the opposite order:  for every rank-$r$ $d\times d$ matrix $\rho$, most sets of $O(rd\poly\log d)$ Pauli measurements can reconstruct that \textit{particular} matrix $\rho$.)

Intuitively, the RIP says that a set of random Pauli measurements is sensitive to all low-rank errors simultaneously. This is important, because it implies stronger error bounds for low-rank matrix recovery \cite{Candes2011}. These bounds show that, when the unknown matrix $\rho$ is \textit{full-rank}, our method returns a (certifiable) rank-$r$ approximation of $\rho$, that is almost as good as the best such approximation (corresponding to the truncated eigenvalue decomposition of $\rho$). 

In Ref.~\cite{Candes2011}, these error bounds were used to show the accuracy of certain compressed sensing estimators, for measurements with additive Gaussian noise.  Here, we use them to upper-bound the sample complexity of our compressed tomography scheme.  (That is, we bound the errors due to estimating each Pauli expectation value from a finite number of experiments.)  Roughly speaking, we show that our scheme uses $O(r^2d^2\log d)$ copies to characterize a rank-$r$ state (up to constant error in trace norm).  When $r = d$, this agrees with the sample complexity of full tomography. Our proof assumes a binomial noise model, but minor modifications could extend this result to other relevant noise models, such as multinomial, Gaussian, or Poissonian noise. 

Furthermore, we show an information-theoretic lower bound for tomography of rank-$r$ states using adaptive sequences of single-copy Pauli measurements: at least $\Omega(r^2 d^2/\log d)$ copies are needed to obtain an estimate with constant accuracy in the trace distance. This generalizes a result from Ref.~\cite{Flammia2011} for pure states. 
Therefore, our upper bound on the sample complexity of compressed tomography is nearly tight, and compressed tomography nearly achieves the optimal sample complexity among all possible methods using Pauli measurements.

Our observation that incomplete sets of observables are often sufficient to unambiguously specify a state gives rise to a new degree of freedom when designing experiments: when aiming to reduce statistical noise in the reconstruction, one can either estimate a small set of observables relatively accurately, or else a large (e.g. complete) set of observables relatively coarsely. Our bounds (as well as our numerics) show that, remarkably, over a very large range of $m$ the \emph{only} quantity relevant for the reconstruction error is $t$, the total number of experiments performed. It does not matter over how many observables the repetitions are distributed. Thus, when fixing $t$ and varying $m$, the reduction in the number of observables and the increase in the number of measurements per observable have no net effect with regard to the fidelity of the estimate, so long as $m \ge \Omega(rd \poly\log d)$.

\textit{Certification.}  We generalize the technique of direct fidelity estimation (DFE)~\cite{Flammia2011, daSilva2011} to work with low-rank states.
% One can combine compressed tomography with direct fidelity estimation (DFE)~\cite{Flammia2011, daSilva2011} to certify an estimate of the density matrix, without making any assumptions about the unknown state.  
Thus, one can use compressed tomography to get an estimated density matrix $\hat{\rho}$, and use DFE to check whether $\hat{\rho}$ agrees with the true state $\rho$. This check is guaranteed to be sound, even if the true state $\rho$ is not approximately low rank. Our extension of DFE may be of more general interest, since it can be used to efficiently certify \textit{any} estimate $\hat\rho$ regardless of whether it was obtained using compressed sensing or not, as long as the rank $r$ of the \textit{estimate} is small (and regardless of the ``true'' rank).

\textit{Numerical simulations.} We compare the performance of several different estimators (methods for reconstructing the unknown density matrix). They include: constrained trace-minimization (a.k.a.\ the matrix Dantzig selector), least squares with trace-norm regularization (a.k.a.\ the matrix Lasso), as well as a standard maximum likelihood estimation (MLE)~\cite{Hradil1997, Banaszek1999, James2001} for comparison. 

We observe that our estimators outperform MLE in essentially all aspects, with the matrix Lasso giving the best results. The fidelity of the estimate is consistently higher using the compressed tomography estimators. Also, the accuracy of the compressed sensing estimates are (as mentioned above) fairly insensitive to the number of measurement settings $m$ (assuming the total time available to run the experiment is fixed). So by choosing $m \ll d^2$, one still obtains accurate estimates, but with much faster classical post-processing, since the size of the data set scales like $O(m)$ rather than $O(d^2)$. 

It may be surprising to the reader that we outperform MLE, since it is often remarked (somewhat vaguely) that ``MLE is optimal.'' However, MLE is a general-purpose method that does not specifically exploit the fact that the state is low-rank. Also, the optimality results for MLE only hold asymptotically and for informationally complete measurements~\cite{Sugiyama2011, Shen2001}; for finite data~\cite{Chakrabarti2009} or for incomplete measurements, MLE can be far from optimal.

From these results, one can extract some lessons about how to use compressed tomography.  Compressed tomography involves two separate ideas: (1) measuring an incomplete set of observables (i.e., choosing $m \ll d^2$), and (2) using trace minimization or regularization to reconstruct low-rank solutions. Usually one does both of these things. Now, suppose the goal is to reconstruct a low-rank state using as few samples as possible. Our results show that one can achieve this goal by doing (2) \textit{without} (1). At the same time, there is no penalty in the quality of the estimate when doing (1), and there are practical reasons for doing it, such as reducing the size of the data set to speed up the classical post-processing.

\textit{Quantum process tomography.} Finally, we adapt our method to perform tomography of processes with small Kraus rank.  Our method is easy to implement, since it requires only the ability to prepare eigenstates of Pauli operators and measure Pauli observables. In particular, we require \emph{no} entangling gates or ancillary systems for the procedure. In contrast to Ref.~\cite{Shabani2011}, our method is not restricted to processes that are element-wise sparse in some known basis, as discussed in Section~\ref{S:process-related-work}. This is an important advantage in practice, because while the ideal (or intended) evolution of a system may be sparse in a known basis, it is often the case that the noise processes perturbing the ideal component are not sparse, and knowledge of these noise processes is key to improving the fidelity of a quantum device with some theoretical ideal.

%------------------------------------------------------------------------------------------------------------%
\subsection{Related Work}
%------------------------------------------------------------------------------------------------------------%

While initial work on tomography considered only linear inversion methods~\cite{Vogel1989}, most subsequent work has largely focused on maximum likelihood methods and to a lesser extent Bayesian methods for state reconstruction~\cite{Jones1991, Hradil1997, Buzek1998, Banaszek1999, Gill2000, Schack2001, James2001, Jezek2003, Neri2005, Tanaka2005, Bagan2006, Audenaert2009, Nunn2010, Blume-Kohout2010a, Vogel1989}.

However, recently there has been a flurry of work which seeks to transcend the standard MLE methods and improve on them in various ways. Our contributions can also be seen in this context. 

One way in which alternatives to MLE are being pursued is through what we call \emph{full rank methods}. Here the idea is somewhat antithetical to ours: the goal is to output a full rank density operator, rather than a rank deficient one. This is desirable in a context where one cannot make the approximation that rare events will never happen. Blume-Kohout's hedged MLE~\cite{Blume-Kohout2010} and Bayesian mean estimation~\cite{Blume-Kohout2010a} are good examples of this type of estimator, as are the minimax estimator of Ref.~\cite{Khoon-Ng2012} and the so-called Max-Ent estimators~\cite{Buzek2004, Teo2011, Teo2011a, Teo2012}. The latter are specifically for the setting where the measurement data are \emph{not} informationally complete, and one tries to minimize the bias of the estimate by maximizing the entropy along the directions where one has no knowledge. 

By contrast, our \emph{low rank} methods do not attempt to reconstruct the complete density matrix, but only a rank-$r$ approximation, which is accurate when the true state is close to low-rank. From this perspective, our methods can be seen as a sort of Occam's Razor, using as few fit parameters as possible while still agreeing with the data~\cite{Yin2011}. Furthermore, as we show here and elsewhere~\cite{Gross2010}, informationally incomplete measurements can still provide faithful state reconstructions up to a small truncation error.

One additional feature of our methods is that we are deeply concerned with the \emph{feasibility} of our estimators for a moderately large number of qubits (say, 10-15). In contrast to most of the existing literature, we adopt the perspective that it is not enough for an estimator to be asymptotically efficient in the number of copies for fixed $d$. We also want the scaling with respect to $d$ to be optimal. We specifically take advantage of the fact that many states and processes are described by low rank objects to reduce this complexity. In this respect, our methods are similar to tomographic protocols that are tailored to special ansatz classes of states, such as those recently developed for use with permutation-invariant states~\cite{Toth2010}, matrix product states~\cite{Cramer2010a} or multi-scale entangled states~\cite{Landon-Cardinal2012}.

Our error bounds are somewhat unique as well. Most prior work on error bounds used either standard resampling techniques or Bayesian methods~\cite{Jones1991, Buzek1998, Schack2001, Tanaka2005, Audenaert2009, Blume-Kohout2010a}. Very recently, Christandl \& Renner and Blume-Kohout independently derived two closely related approaches for obtaining confidence regions that satisfy or nearly satisfy certain optimality criteria~\cite{Christandl2011, Blume-Kohout2012}. Especially these latter approaches can give very tight error bounds on an estimate, but they can be computationally challenging to implement for large systems. The error bounds which most closely resemble ours are of the ``large deviation type''; see for example the discussion in Ref.~\cite{Sugiyama2011}. This is true for the new improved error bounds, as well as the original bounds proven in Refs.~\cite{Gross2010,Gross2011}. These types of bounds are much easier to calculate in practice, which agrees with our philosophy on computational complexity, but may be somewhat looser than the optimal error bounds obtainable through other more computationally intensive methods such as those of Refs.~\cite{Christandl2011, Blume-Kohout2012}.

%------------------------------------------------------------------------------------------------------------%
\subsection{Notation and Outline}
%------------------------------------------------------------------------------------------------------------%

We denote Pauli operators by $P$ or $P_i$. We define $[n] = \{1,\ldots,n\}$. The norms we use are the standard Euclidean vector norm $\|x\|_2$, the Frobenius norm $\|X\|_F = \sqrt{\Tr(X^\dagger X)}$, the operator norm $\|X\| = \sqrt{\lambda_{\max}(X^\dagger X)}$ and the trace norm $\|X\|_{\tr} = \Tr\abs{X}$, where $\abs{X} = \sqrt{X^\dagger X}$. The unknown ``true'' state is denoted $\rho$ and any estimators for $\rho$ are given a hat: $\hat\rho$. The expectation value of a random variable $X$ is denoted $\EE X$. We denote by $\HH^d$ the set of $d \times d$ Hermitian matrices.

The paper is organized as follows. In Section~\ref{sec-theory} we detail the estimators and error bounds, then upper bound the sample complexity. In Section~\ref{S:lowerbound} we derive lower bounds on the sample complexity. In Section~\ref{S:cert} we find an efficient method of certifying the state estimate. 
In Section~\ref{S:numerics} we detail our numerical investigations. We show how our scheme can be applied to quantum channels in Section~\ref{S:process} and conclude in Section~\ref{S:conclusion}.

% !TEX root = ../tradeoff.tex

%------------------------------------------------------------------------------------------------------------%
\section{Theory} \label{sec-theory}
%------------------------------------------------------------------------------------------------------------%

We describe our compressed tomography scheme in detail.  First we describe the measurement procedure, and the method for reconstructing the density matrix.  Then we prove error bounds and analyze the sample complexity.

%------------------------------------------------------------------------------------------------------------%
\subsection{Random Pauli Measurements}
%------------------------------------------------------------------------------------------------------------%

Consider a system of $n$ qubits, and let $d = 2^n$.  Let $\calP$ be the set of all $d^2$ Pauli operators, i.e., matrices of the form $P = \sigma_1 \ox \cdots \ox \sigma_n$ where $\sigma_i \in \set{I,\sigma^x,\sigma^y,\sigma^z}$.  

Our tomography scheme works as follows.  First, choose $m$ Pauli operators, $P_1,\ldots,P_m$, by sampling independently and uniformly at random from $\calP$.  (Alternatively, one can choose these Pauli operators randomly without replacement~\cite{Gross2010a}, but independent sampling greatly simplifies the analysis.) We will use $t$ copies of the unknown quantum state $\rho$.  For each $i \in [m]$, take $t/m$ copies of the state $\rho$, measure the Pauli observable $P_i$ on each one, and average the measurement outcomes to obtain an estimate of the expectation value $\Tr(P_i\rho)$.  (We will discuss how to set $m$ and $t$ later.  Intuitively, to learn a $d\times d$ density matrix with rank $r$, we will set $m \sim rd\log^6 d$ and $t \sim r^2d^2\log d$.)  

To state this more concisely, we introduce some notation.  Define the \emph{sampling operator} to be a linear map $\calA:\: \HH^d \rightarrow \RR^m$ defined for all $i \in [m]$ by
\begin{equation}\label{eqn-samp-op}
	(\calA(\rho))_i = \sqrt{\tfrac{d}{m}} \Tr(P_i\rho) \,.
\end{equation}
The normalization is chosen so that $\EE \calA^*\calA = \calI$, where $\calI$ denotes the identity superoperator and $\calA^*$ is the adjoint of $\calA$.  We can then write the output of our measurement procedure as a vector 
\begin{align}
	y = \calA(\rho) + z\,,
\end{align}
where $z$ represents statistical noise due to the finite number of samples, or even due to an adversary.

%------------------------------------------------------------------------------------------------------------%
\subsection{Reconstructing the Density Matrix}
%------------------------------------------------------------------------------------------------------------%

We now show two methods for estimating the density matrix $\rho$.  Both are based on the same intuition:  find a matrix $X \in \HH^d$ that fits the data $y$ while minimizing the trace norm $\norm{X}_\tr$, which serves as a surrogate for minimizing the rank of $X$. In both cases, this amounts to a convex program, which can be solved efficiently.

(We mention that at this point we do not require that our density operators are normalized to have unit trace. We will return to this point later in Section~\ref{S:numerics}.)

The first estimator is obtained by constrained trace-minimization (a.k.a.\ the matrix Dantzig selector):
\begin{equation}\label{eqn-ds}
	\hat{\rho}_\ds = \arg\min_X \norm{X}_\tr \st \norm{\calA^*(\calA(X)-y)} \leq \lambda\,.
\end{equation}
The parameter $\lambda$ should be set according to the amount of noise in the data $y$; we will discuss this later.

The second estimator is obtained by least-squares linear regression with trace-norm regularization (a.k.a.\ the matrix Lasso):
\begin{equation}\label{eqn-lasso}
	\hat{\rho}_\lasso = \arg\min_X \tfrac{1}{2} \norm{\calA(X)-y}_2^2 + \mu \norm{X}_\tr\,.
\end{equation}
Again the regularization parameter $\mu$ should be set according to the noise level; we will discuss this later.

One additional point is that we do not require positivity of the output in our definition of the estimators (\ref{eqn-ds}), (\ref{eqn-lasso}). One can add this constraint (since it is convex) and the conclusions below remain unaltered. We will explicitly add positivity later on when we do numerical simulations, and discuss any tradeoffs that result from this.

%------------------------------------------------------------------------------------------------------------%
\subsection{Error Bounds}
%------------------------------------------------------------------------------------------------------------%

In previous work on compressed tomography \cite{Gross2010, Gross2011}, error bounds were proved by constructing a ``dual certificate'' \cite{Candes2009} (using convex duality to characterize the solution to the trace-minimization convex program).  Here we derive stronger bounds using a different tool, known as the restricted isometry property (RIP).  The RIP for low-rank matrices was first introduced in Ref.~\cite{Recht2007}, and was recently shown to hold for random Pauli measurements \cite{Liu2011}.  

We say that the sampling operator $\calA$ satisfies the rank-$r$ restricted isometry property if there exists some constant $0 \leq \delta_r < 1$ such that, for all $X \in \CC^{d\times d}$ with rank~$r$, 
\begin{equation}
	(1-\delta_r) \norm{X}_F \leq \norm{\calA(X)}_2 \leq (1+\delta_r) \norm{X}_F\,.
\end{equation}
For our purposes, we can assume that $X$ is Hermitian.  Note that this notion of RIP is analogous to the one used in Ref.~\cite{Shabani2011}, except that it holds over the set of low-rank matrices, rather than the set of sparse matrices.

The random Pauli sampling operator (\ref{eqn-samp-op}) satisfies RIP with high probability, provided that $m \geq Crd\log^6 d$ (for some absolute constant $C$); this was recently shown in Ref.~\cite[Theorem 2.1]{Liu2011}. We note, however, that this RIP result requires $m$ to be larger by a $\poly\log d$ factor compared to previous results based on dual certificates. Although $m$ is slightly larger, the advantage is that when combined with the results of Ref.~\cite{Candes2011}, this immediately implies strong error bounds for the matrix Dantzig selector and the matrix Lasso. 

To state these improved error bounds precisely, we introduce some definitions. For the rest of Section \ref{sec-theory}, let $C$, $C_0$, $C_1$, $C'_0$ and $C'_1$ be fixed absolute constants. For any quantum state $\rho$, we write $\rho = \rho_r + \rho_c$, where $\rho_r$ is the best rank-$r$ approximation to $\rho$ (consisting of the largest $r$ eigenvalues and eigenvectors), and $\rho_c$ is the residual part. Now we have the following:

\begin{theorem}
\label{thm-errorbound}
Let $\calA$ be the random Pauli sampling operator (\ref{eqn-samp-op}) with $m \geq Crd\log^6 d$. Then, with high probability, the following holds:

Let $\hat{\rho}_\ds$ be the matrix Dantzig selector (\ref{eqn-ds}), and choose $\lambda$ so that $\norm{\calA^*(z)} \leq \lambda$.  Then 
\begin{equation*}
	\norm{\hat{\rho}_\ds - \rho}_\tr \leq C_0 r\lambda + C_1 \norm{\rho_c}_\tr\,.
\end{equation*}

Alternatively, let $\hat{\rho}_\lasso$ be the matrix Lasso (\ref{eqn-lasso}), and choose $\mu$ so that $\norm{\calA^*(z)} \leq \mu/2$.  Then 
\begin{equation*}
	\norm{\hat{\rho}_{\lasso} - \rho}_\tr \leq C'_0 r\mu + C'_1 \norm{\rho_c}_\tr\,.
\end{equation*}
\end{theorem}

In these error bounds, the first term depends on the statistical noise $z$.  This in turn depends on the number of copies of the state that are available in the experiment; we will discuss this in the next section.  The second term is the rank-$r$ approximation error.  It is clearly optimal, up to a constant factor.

\textbf{Proof:} These error bounds follow from the RIP as shown by Theorem 2.1 in Ref.~\cite{Liu2011}, and a straightforward modification of Lemma 3.2 in Ref.~\cite{Candes2011} to bound the error in trace norm rather than Frobenius norm (this is similar to the proof of Theorem 5 in Ref.~\cite{Fazel2008}).  

The modification of Lemma 3.2 in Ref.~\cite{Candes2011} is as follows\footnote{Note that Ref.~\cite{Candes2011} contains a typo in Lemma 3.2:  on the right hand side, in the second term, it should be $\norm{M_c}_*/\sqrt{r}$, not $\norm{M_c}_*/r$.}.  (For the remainder of this section, equation numbers of the form (III.x) refer to Ref.~\cite{Candes2011}.)  In the case of the Dantzig selector, let $H = \hat{\rho}_\ds - \rho$.  Following equation (III.8), we can get the following bound:  
\begin{equation}\label{eqn-ds-pf-1}
	\norm{H}_\tr \leq \norm{H_0}_\tr + \norm{H_c}_\tr \leq 2\norm{H_0}_\tr + 2\norm{\rho_c}_\tr, 
\end{equation}
where we used the triangle inequality and equation (III.8).  Then, at the end of the proof, we write:
\begin{equation}\label{eqn-ds-pf-2}
	\begin{split}
	\norm{H_0}_\tr &\leq \sqrt{2r} \norm{H_0}_F \leq \sqrt{2r} \norm{H_0+H_1}_F \\
	&\leq C_1 4\sqrt{2} r \lambda + C_1 2\sqrt{2} \delta_{4r} \norm{\rho_c}_\tr, 
	\end{split}
\end{equation}
where we used Cauchy-Schwarz, the fact that $H_0$ and $H_1$ are orthogonal, the bound on $\norm{H_0+H_1}_F$ following equation (III.13), and equation (III.7).  Combining (\ref{eqn-ds-pf-1}) and (\ref{eqn-ds-pf-2}) gives our desired error bound.  The error bound for the Lasso is obtained in a similar way; see section III.G in Ref.~\cite{Candes2011}.  $\square$

%------------------------------------------------------------------------------------------------------------%
\subsection{Sample Complexity}
%------------------------------------------------------------------------------------------------------------%

Here we bound the sample complexity of our tomography scheme, that is, we bound the number of copies of the unknown quantum state $\rho$ that are needed to obtain our estimate up to some accuracy. What we show, roughly speaking, is that $t = O\bigl((\tfrac{rd}{\varepsilon})^2 \log d\bigr)$ copies are sufficient to reconstruct an estimate of an unknown rank $r$ state up to accuracy $\varepsilon$ in the trace distance. For comparison, note that when $r = d$, and one does full tomography, $O(d^4/\varepsilon^2)$ copies are sufficient to estimate a full-rank state with accuracy $\varepsilon$ in trace distance\footnote{To see this, let $P_1,\ldots,P_{d^2}$ denote the Pauli matrices. For each $i \in [d^2]$, measure $P_i$ $O(d^2/\varepsilon^2)$ times, to estimate its expectation value with additive error $\pm O(\varepsilon/d)$. Equivalently, one estimates the expectation value of $P_i/\sqrt{d}$ with additive error $\pm O(\varepsilon/d^{3/2})$. Using linear inversion, one gets an estimated density matrix with additive error $O(\varepsilon/d^{1/2})$ in Frobenius norm, which implies error $O(\varepsilon)$ in trace norm.}.

To make this claim precise, we need to specify how we construct our data vector $y$ from the measurement outcomes on the $t$ copies of the state $\rho$. For the matrix Dantzig selector, suppose that 
\begin{equation}\label{E:copiest}
t \geq 2 C_4 (C_0 r/\varepsilon)^2 d(d+1) \log d
\end{equation}
for some constants $C_4 > 1$ and $\varepsilon \leq C_0$. (For the matrix Lasso, substitute $C'_0$ for $C_0$ in these equations.)
We construct an estimate of $\calA(\rho)$ as follows: for each $i \in [m]$, we take $t/m$ copies of $\rho$, measure the random Pauli observable $P_i$ on each of the copies, and use this to estimate $\Tr(P_i\rho)$.  Then let $y$ be the resulting estimate of $\calA(\rho)$, and let $z = y - \calA(\rho)$. Everything else is defined exactly as in Theorem \ref{thm-errorbound}.

\begin{theorem}\label{thm-samplecomplexity}
Given $t = O\bigl((\tfrac{rd}{\varepsilon})^2 \log d\bigr)$ copies of $\rho$ as in Eq.~(\ref{E:copiest}) and measured as discussed above, then the following holds with high probability over the measurement outcomes:

Let $\hat{\rho}_\ds$ be the matrix Dantzig selector (\ref{eqn-ds}), and set $\lambda = \varepsilon/(C_0 r)$ for some $\varepsilon > 0$.  Then 
\begin{equation*}
	\norm{\hat{\rho}_\ds - \rho}_\tr \leq \varepsilon + C_1 \norm{\rho_c}_\tr\,.
\end{equation*}

Alternatively, let $\hat{\rho}_\lasso$ be the matrix Lasso (\ref{eqn-lasso}), and set $\mu = \varepsilon/(C'_0 r)$.  Then 
\begin{equation*}
	\norm{\hat{\rho}_\lasso - \rho}_\tr \leq \varepsilon + C'_1 \norm{\rho_c}_\tr\,.
\end{equation*}
\end{theorem}

\textbf{Proof:} Our claim reduces to the following question:  if we fix some value of $\lambda > 0$, how many copies of $\rho$ are needed to ensure that the measurement data $y$ satisfies $\norm{\calA^*(y-\calA(\rho))} \leq \lambda$?  Then one can apply Theorem \ref{thm-errorbound} to get an error bound for our estimate of $\rho$.

Let $t$ be the number of copies of $\rho$.  Say we fix the measurement operator $\calA$, i.e., we fix the choice of the Pauli observables $P_1,\ldots,P_m$.  (The measurement outcomes are still random, however.)  For $i \in [m]$ and $j \in [t/m]$, let $B_{ij} \in \set{1,-1}$ be the outcome of the $j$'th repetition of the experiment that measures the $i$'th Pauli observable $P_i$.  Note that $\EE B_{ij} = \Tr(P_i\rho)$.  Then construct the vector $y \in \RR^m$ containing the estimated expectation values (scaled by $\sqrt{d/m}$):
\begin{equation}
	y_i = \sqrt{\tfrac{d}{m}} \cdot \tfrac{m}{t} \sum_{j=1}^{t/m} B_{ij}, \quad i \in [m].
\end{equation}
Note that $\EE y = \calA(\rho)$.  

We will bound the deviation $\norm{\calA^*(y-\calA(\rho))}$, using the matrix Bernstein inequality.  First we write 
\begin{equation}
	\calA^*(y) = \sqrt{\tfrac{d}{m}} \sum_{i=1}^m P_i y_i
	 = \tfrac{d}{t} \sum_{i=1}^m \sum_{j=1}^{t/m} P_i B_{ij}, 
\end{equation}
and also
\begin{equation}
	\calA^*\calA(\rho) = \tfrac{d}{m} \sum_{i=1}^m P_i \Tr(P_i\rho).
\end{equation}
We can now write $\calA^*(y-\calA(\rho))$ as a sum of independent (but not identical) matrix-valued random variables:
\begin{equation}
	\calA^*(y-\calA(\rho)) = \sum_{i=1}^m \sum_{j=1}^{t/m} X_{ij}, \quad
	X_{ij} = \tfrac{d}{t} P_i \bigl[B_{ij} - \Tr(P_i\rho)\bigr].
\end{equation}
Note that $\EE X_{ij} = 0$ and $\norm{X_{ij}} \leq 2d/t =: R$.  Also, for the second moment we have 
\begin{equation}
	\begin{split}
	\EE\bigl(X_{ij}^2\bigr) &= \EE\Bigl(\tfrac{d^2}{t^2} I \bigl[B_{ij} - \Tr(P_i\rho)\bigr]^2\Bigr) \\
	 &= \tfrac{d^2}{t^2} I \bigl[1 - \Tr(P_i\rho)^2\bigr].
	\end{split}
\end{equation}
Then we have
\begin{equation}
	\begin{split}
	\sigma^2
	 &:= \Bigl\lVert \sum_{ij} \EE(X_{ij}^2) \Bigr\rVert = \sum_{ij} \tfrac{d^2}{t^2} \bigl[1 - \Tr(P_i\rho)^2\bigr] \\
	 &\leq t \cdot \tfrac{d^2}{t^2} = \tfrac{d^2}{t}.
	\end{split}
\end{equation}
Now the matrix Bernstein inequality (Theorem 1.4 in \cite{Tropp2010}) implies that 
\begin{equation}
	\begin{split}
	\Pr\bigl[\norm{\calA^*&(y-\calA(\rho))} \geq \lambda\bigr]
	 \leq d\cdot \exp\Bigl( -\tfrac{\lambda^2/2}{\sigma^2 + (R\lambda/3)} \Bigr) \\
	 &\leq d\cdot \exp\Bigl( -\tfrac{t\lambda^2/2}{d(d+1)} \Bigr)
	\end{split}
\end{equation}
(where we assumed $\lambda \leq 1$). 

For the matrix Dantzig selector, we set $\lambda = \varepsilon/(C_0 r)$, and we get that, for any $t \geq 2 C_4 \lambda^{-2} d(d+1) \log d = 2 C_4 (C_0 r/\varepsilon)^2 d(d+1) \log d$, 
\begin{equation}
	\begin{split}
	\Pr\bigl[\norm{\calA^*(y-\calA(\rho))} \geq \tfrac{\varepsilon}{C_0 r}\bigr]
	 &\leq d\cdot \exp(-C_4 \log d) \\
	 &= d^{1-C_4},
	\end{split}
\end{equation}
which is exponentially small in $C_4$.  Plugging into Theorem \ref{thm-errorbound} completes the proof of our claim.  A similar argument works for the matrix Lasso.  $\square$

% !TEX root = ../tradeoff.tex

%------------------------------------------------------------------------------------------------------------%
\section{Lower Bounds}\label{S:lowerbound}
%------------------------------------------------------------------------------------------------------------%

How good are the sample complexities of our algorithms? Here we go a long way toward answering this question by proving nearly tight lower bounds on the sample complexity for generic rank $r$ quantum states using single-copy Pauli measurements. Previous work on single-copy lower bounds has only treated the case of pure states~\cite{Flammia2011}. 

Roughly speaking, we show the following, which we will make precise at the end of the section.
\begin{theorem}[Imprecise formulation]\label{Thm:lower}
	The number of copies $t$ needs to grow at least as fast as $\Omega\bigl(r^2d^2/\log
	d\bigr)$ to guarantee a constant trace-norm confidence interval for all
	rank-$r$ states.
\end{theorem}

The argument proceeds in three steps. First, we fix our notion of risk to be the minimax risk, meaning we seek to minimize our worst case error according to some error metric such as the trace distance. We want to know how many copies of the unknown state we need to make this minimax risk an arbitrarily small constant. For a fixed set of two-outcome measurements, say Pauli measurements, we then show that some states require many copies to achieve this. In particular, these states have the property that they are globally distinguishable (their trace distance is bounded from below by a constant) but their (Pauli) measurement statistics look approximately the same (each measurement outcome is close to unbiased). The more such states there are, the more copies we need to distinguish between them all using solely Pauli measurements. Finally, we use a randomized argument to show that in fact there are exponentially many such states. This yields the desired lower bound on the sample complexity.

Let $\Sigma$ be some set of density operators. We want to put lower bounds on the performance of any estimation protocol for states in $\Sigma$. (We do not initially restrict ourselves to states with low rank.)

We assume the protocol has access to $t$ copies of an unknown state $\rho\in\Sigma$ on which it performs measurements one by one. At the $i$th step, it has to decide which observable to measure. Let us  restrict ourselves to binary POVM measurements $\{\Pi_i, \Id-\Pi_i\}$, where each $\Pi_i$ satisfies $0\le \Pi_i \le \Id$ and may be chosen from a set $\mathcal{P}$. (We do not initially restrict ourselves to Pauli measurements, either.) We allow the choice of the $i$th observable to depend on the previous outcomes. We refer to the random variables which jointly describe the choice of the $i$th measurement and its random outcome as $Y_i$. At the end, these are mapped to an estimate $\hat\rho(Y_1, \dots, Y_t)\in \Sigma$.

In other words, an estimation protocol is specified by a set of functions
\begin{align*}
	\Pi_i&: Y_1, \dots, Y_{i-1} \mapsto \mathcal{P}\,, \\
	\hat\rho&: Y_1, \dots, Y_t \mapsto \Sigma\,.
\end{align*}

Consider a distance measure $\Delta:\Sigma\times\Sigma \to \RR$ on the states in $\Sigma$. (For example, this could be the trace distance or the infidelity; we need not specify.) Suppose that the maximum deviation we wish to tolerate between our unknown state and the estimate is given by $\eps>0$. Now define the minimax risk
\begin{align}\label{eqn:protocol}
	M^*(\eps) = \inf_{\langle \hat \rho, \Pi_i \rangle} 
	\sup_{\rho\in\Sigma} \Pr\bigl[ \Delta(\hat\rho(Y), \rho) > \eps\bigr],
\end{align}
where the infimum is over all estimation procedures $\langle \hat
\rho, \Pi_i\rangle$ on $t$ copies with estimator $\hat\rho$ and
choice of observables given by $\Pi_i$. That is, we are considering the ``best'' protocol to be the one whose worst-case probability of deviation is the least. 

The next lemma shows that if there are a large number of states in $\Sigma$ which are far apart (by at least $\eps$), and whose statistics look nearly random for all measurements in $\P$, then the number of copies $t$ must be large too to avoid having a large minimax risk.

\begin{lemma}\label{minimaxlemma}
%	Let $\Sigma$ be a set of density operators. Let 
%	\begin{equation*}
%		d: \Sigma \times \Sigma \to \RR
%	\end{equation*}
%	be some distance measure on $\Sigma$.
%	Let $\mathcal{P}$ be a set of POVM elements: 
%	\begin{equation*}
%		\forall\, \Pi\in\mathcal{P}: \> 0 \leq \Pi \leq \Id.
%	\end{equation*}
	Assume there are states $\rho_1, \dots, \rho_s\subset\Sigma$ such
	that 
	\begin{align*}
		&\forall\>i\neq j: \> \Delta(\rho_i, \rho_j) \geq \eps \,,\\
		&\forall\>i,\forall\, \Pi \in \mathcal{P}:\> \bigl|\tr \rho_i \Pi - \tfrac12\bigr| \leq \alpha \,.
	\end{align*}
	Then the minimax risk as defined in (\ref{eqn:protocol}) fulfills
	\begin{align*}
		M^*(\eps)>1-\frac{4\alpha^2 t + 1}{\log s}.
	\end{align*}
\end{lemma}

\begin{proof}
	Let $X$ be a random variable taking values uniformly in
	$[s]$. Let $Y_1, \dots, Y_t$ be random variables, $Y_i$
	describing the outcome of the $i$th measurement performed on
	$\rho_X$. Define 
	\begin{equation*}
		\hat X(Y) := \arg \min_i \Delta(\hat\rho(Y),\rho_i).
	\end{equation*}
	Then 
	\begin{equation}\label{eqn:lb1}
		\Pr\bigl[ \Delta(\hat\rho(Y),\rho_i) > \eps \bigr] 
		\geq
		\Pr\bigl[ \hat X(Y) \neq X \bigr] \,.
	\end{equation}

	Now combine Fano's inequality
	\begin{align*}
		H(X|\hat X)\leq	1+
		\Pr[ \hat X \neq X ] \log s,
	\end{align*}
	the data processing inequality 
	\begin{equation*}
		%H(X|\hat X)\geq H(X|Y),
		I(X;\hat X(Y))\leq I(X;Y),
	\end{equation*}
	in terms of the mutual information $I(X;Y):=H(X)-H(X|Y)$,
	and the fact that $H(X)=\log s$ to get
	\begin{align*}
		\Pr[ \hat X(Y) \neq X ]
		&\geq
		\frac{H(X|\hat X) - 1}{\log s} \\
		&=
		\frac{H(X)-I(X;\hat X) - 1}{\log s} \\
		&=
		1-\frac{I(X;\hat X) + 1}{\log s} \\
		&\geq
		1-\frac{I(X;Y) + 1}{\log s} \\
		&=
		1-\frac{ 
		H(Y) - H(Y|X) + 1}{\log s} \\
		&\geq
		1-\frac{ 
		t - \frac1s \sum_{i=1}^s H(Y|X=i) + 1}{\log s}.
	\end{align*}

	Let $h(p)$ be the binary entropy and recall the standard estimate
	\begin{equation*}
		 h(1/2\pm\alpha) \geq (1-4\alpha^2)\,.
	\end{equation*}
	Combine that with the chain rule \cite[Theorem~2.5.1]{Cover1991}:
	\begin{align*}
		H(Y|X=i) &=
		\sum_{j=1}^t H(Y_j | Y_{j-1}, \dots, Y_1, X=i) \\
		&\geq
		t (1-4\alpha^2).
	\end{align*}
	The advertised bound follows.
\end{proof}

By applying the following lemma $s < \e^{crd}$ times, we can randomly create a set of $s$ states each with rank $r$ that satisfy the conditions of Lemma~\ref{minimaxlemma} in terms of the trace distance and the set of Pauli measurements. 

\begin{lemma}\label{randomized}
	For any $0< \eps < 1-\frac{r}{d}$, let $\rho_1, \dots, \rho_s$ be normalized\footnote{Normalized to have trace 1.} rank-$r$ projections on $\CC^d$, where $s< \e^{c(\eps) rd}$ and $c(\eps)$ is specified below. Then there exists a normalized rank-$r$ projection $\rho$ such that:
	\begin{eqnarray}\label{eqn:farApart}
		&\forall\>i \in [s]:\> \frac{1}{2}\|\rho - \rho_i\|_{\tr} \geq \eps \,,\\
		\label{eqn:smallPauli}
		&\forall\>P_k \neq \Id:\> \bigl|\tr\bigl[\frac12(\Id\pm P_k) \rho\bigr]-\frac12\bigr| \leq \alpha \,.
	\end{eqnarray}
	Here, $\alpha^2 = O\bigr(\tfrac{\log d}{rd}\bigr)$, the $P_k$ are $n$-qubit Pauli operators, and 
	\begin{align*}
		c(\eps) = \frac{\ln(8/\pi)}{2 r d}+ \frac{1}{32}\Bigl[\bigl(1-\tfrac{r}{d}\bigr)-\eps\Bigr]^2\,.
	\end{align*}
\end{lemma}

\begin{proof}
	Let $\rho_0$ be some normalized rank-$r$ projection and choose
	$\rho$ according to\footnote{For the duration of this proof, the letter $O$ denotes an element of $\mathsf{SO}(d)$ instead of the asymptotic big-$O$ notation.}  
	\begin{equation*}
		\rho = O \rho_0 O^T
	\end{equation*}
	for a Haar-random $O\in \textsf{SO}(d)$. Here we use the special orthogonal group $\textsf{SO}(d)$ because the analysis becomes marginally simpler than if we use a unitary group. 

	To check (\ref{eqn:farApart}), choose $i \in [s]$ and define
	$R_i$ to be the projector onto the range of $\rho_i$. Also define the function 
	\begin{equation*}
		f: O \mapsto \|\rho_i - O \rho_0 O^T \|_{\tr}\,.
	\end{equation*}
	We can bound the magnitude of $f$ using the pinching inequality:
	\begin{align*}
		f(O) \geq& \|\rho_i - R_i O \rho_0 O^T R_i\|_{\tr} + \|R_i^\perp O  \rho_0 O^T R_i^\perp\|_{\tr} \\
           \geq& \tr(\rho_i) -\tr(R_i O \rho_0 O^T R_i) + \tr(R_i^\perp O  \rho_0 O^T R_i^\perp) \\
		 %=& 1  -\tr R_i O \rho_0 O^T + \tr R_i^\perp O \rho_0 O^T.
		 =& 1  +\tr\bigl[(R_i^\perp - R_i)O \rho_0 O^T\bigr] \,.
	\end{align*}
	From this we can bound the expectation value of $f$ over the special orthogonal group:
	\begin{align*}
		\EE\bigl[f(O)\bigr]
		\geq& 1+ \tr\bigl[(R_i^\perp - R_i) \big( \EE O \rho O^T \big)\bigr]  \\
		=& 1+ \frac{1}{d} \tr\bigl[(R_i^\perp - R_i) \Id\bigr] \\
		=& 1 + \frac{d-2r}{d}=2\Bigl(1-\frac{r}{d}\Bigr)\,.
	\end{align*}
	Next we get an upper bound of $\frac{4}{\sqrt{r}}$ on the Lipschitz constant of $f$ with respect to the Frobenius norm:
	\begin{align*}
	\begin{split}
		|f(O+\Delta) & - f(O)| \\
		\leq& 
		\|(O+\Delta) \rho_0 (O+\Delta)^T - O \rho_0 O^T\|_{\tr} \\
		\leq& 
		\|O \rho_0 \Delta^T\|_{\tr} + \|\Delta\rho_0
		O^T\|_{\tr} + \|\Delta \rho_0 \Delta^T\|_{\tr} \\
		=& 
		2\|\Delta \rho_0\|_{\tr}  +
		\tr(\Delta \rho_0 \Delta^T) \\
		\leq& 
		2\sqrt r \|\Delta \rho_0\|_F + \tr(\rho_0 \Delta \Delta^T) \\
		\leq& 
		2\sqrt r \|\Delta\|_F \|\rho_0\| + \frac{1}{r} \|\Delta\| \|\Delta^T\|_{\tr}  \\
		\leq &
		\frac{2}{\sqrt r} \|\Delta\|_F + \frac{2}{r} \sqrt{r} \|\Delta\|_F
		\leq 
		\frac4{\sqrt r} \|\Delta\|_F \,,
	\end{split}
	\end{align*}
	where we use $\|\Delta\| \le 2$ in the last line, which follows from the triangle inequality and the fact that any $\Delta$ can be written as a difference $\Delta = O'-O$ for $O' \in \mathsf{SO}(d)$.
	
	From these ingredients, we can invoke L\'evy's Lemma on the special orthogonal group \cite[Theorem 6.5.1]{Milman1986} to get that for all $t>0$,
	\begin{align*}
		\Pr\bigl[ \| \rho_i - \rho \|_{\tr} < 2(1-r/d)- t \bigr]
		\leq
		\e^{c_1} \exp\Bigl(- \frac{c_2 t^2 rd}{16}\Bigr),
	\end{align*}
	where the constants are given by $c_1 = \ln\bigl(\sqrt{\pi/8}\bigr)$ and $c_2 = 1/8$.
	Now choose $t = 2(1-r/d) - 2\eps$ and apply the union bound to obtain
	\begin{align*}
		\Pr\bigl[ (\ref{eqn:farApart}) & \text{ does not hold}\bigr] \\ 
		<\, & \e^{c(\eps) rd} \Pr\bigl[ \| \rho_i - \rho \|_{\tr} < 2(1-r/d)- t  \bigr]\\
		\leq & \exp\mathclose{}\bigg(rd\Bigl[c(\eps)-\tfrac{c_2 t^2 }{16}\Big]+ c_1\bigg)=1.
	\end{align*}
	The upper bound on $\eps$ follows from the requirement that $t > 0$. This shows that $\rho$ indeed satisfies Eq.~(\ref{eqn:farApart}).

	Now we move on to Eq.~(\ref{eqn:smallPauli}). For any non-identity Pauli
	matrix $P_k$, define a function
	\begin{equation*}
		f: O \mapsto \tr(P_k O \rho_0 O^T) \,.
	\end{equation*}
	Clearly, we have $\EE[f(O)] = 0$. We again wish to bound the rate of change so that we can use L\'evy's Lemma, so we compute
	\begin{align*}
		(\mathrm{d}_O f)(\Delta)  
		=& \tr\bigl(\rho_0 O^T P_k  \Delta\bigr) + \tr\bigl(P_k O \rho_0 \Delta^T\bigr) \\
		=& \tr\bigl[(\rho_0 O^T P_k + \rho_0^T O^T P_k^T) \Delta \bigr]\,,
	\end{align*}
	which implies that
	\begin{equation*}
		\|\nabla f(O)\|_2 =  
		\|\rho_0 O^T P_k + \rho_0^T O^T P_k^T \|_F 
		\leq \frac{2}{\sqrt{r}} \,. 
	\end{equation*}
	L\'evy's Lemma then gives for all $t>0$
	\begin{equation*}
		\Pr\bigl[ \abs{ \tr P_k \rho } >  t \bigr]
		\leq
		\e^{c_1} \exp\biggl(- \frac{c_2 t^2 rd}{4}\biggr).
	\end{equation*}
	Choosing $t = 2\alpha$, and $\alpha^2 = 4\ln(d^4 \pi/8)/(rd)$, then the union bound gives us
	\begin{align*}
		\Pr\bigl[ (\ref{eqn:smallPauli}) & \text{ does not hold}\bigr] \\
		< &\, d^2 \Pr\bigl[ |\tr P_k \rho| > 2\alpha\bigr] \\
		\leq& \exp\Bigl(2\ln d -\frac{c_2(2\alpha)^2rd}{4}+ c_1\Bigr) = 1\,,
	\end{align*}
	from which the lemma follows.
	\end{proof}

We remark that a version of Lemma~\ref{randomized} continues to hold even if we can adaptively choose from as many as $2^{O(n)}$ additional measurements which are globally unitarily equivalent to Pauli measurements. 

Combining the two previous lemmas yields a precise formulation of Theorem \ref{Thm:lower}.
\begin{theorem}[Precise version of Theorem~\ref{Thm:lower}]
	Fix $\eps \in (0, 1-\tfrac{r}{d})$ and $\delta \in [0,1)$. Then for our bound to allow for $M^*(\eps) \le \delta$ we must have that the number of copies $t$ of $\rho$ grows like
	\begin{equation*}
		t = \Omega\biggl(\frac{r^2 d^2}{\log d}\biggr)\,,
	\end{equation*}
	where the implicit constant depends on $\delta$ and $\eps$.
\end{theorem}

% !TEX root = ../tradeoff.tex

%------------------------------------------------------------------------------------------------------------%
\section{Certifying the State Estimate}\label{S:cert}
%------------------------------------------------------------------------------------------------------------%

Here we sketch how the technique of \emph{direct fidelity estimation} (DFE), introduced in Refs.~\cite{Flammia2011, daSilva2011} for pure states, can be used to estimate the fidelity between a low rank estimate $\hat\rho$ and the true state $\rho$. The only assumption we make is that $\hat\rho$ is a positive semidefinite matrix with $\Tr(\hat\rho) \le 1$ and $\rank(\hat\rho) = r$. No assumption at all is needed on $\rho$. In fact, we also do not assume that we obtained the estimate $\hat\rho$ from any of the estimators which were discussed previously. Our certification procedure works regardless of how one obtains $\hat\rho$, and so it even applies to the situation where $\hat\rho$ was chosen from a subset of variational ansatz states, as in \cite{Cramer2010a}.

Recall that the main idea of DFE is to take a known pure state $\proj\psi$ and from it define a probability distribution $\Pr(i)$ such that by estimating the Pauli expectation values $\Tr(\rho P_i)$ and suitably averaging it over $\Pr(i)$ we can learn an estimate of  $\bra\psi\rho\ket\psi$. In fact, one does not need to do a full average; simply sampling from the distribution a few times is sufficient to produce a good estimate. 

For non-Hermitian rank-$1$ matrices $\ketbra{\phi_j}{\phi_k}$ instead of pure states, we need a very slight modification of DFE. Following the notation in Ref.~\cite{Flammia2011}, we simply redefine the probability distribution as $\Pr(i) = \abs{\bra{\phi_j}P_i\ket{\phi_k}}^2/d$ and the random variable $X = \Tr(P_i \rho)/\bra{\phi_k}P_i\ket{\phi_j}$. It is easy to check that $\E(X) = \bra{\phi_j}\rho\ket{\phi_k}$ and the variance of $X$ is at most one. Then all of the conclusions in Refs.~\cite{Flammia2011, daSilva2011} hold for estimating the quantity $\bra{\phi_j}\rho\ket{\phi_k}$. In particular, we can obtain an estimate to within an error $\pm \eps$ with probability $1-\delta$ by using only single-copy Pauli measurements and $O(\frac{1}{\eps^2\delta} + \frac{d \log(1/\delta)}{\eps^2})$ copies of $\rho$.

Our next result shows that by obtaining several such estimates using DFE, we can also infer an estimate of the mixed state fidelity between an unknown state $\rho$ and a low rank estimate $\hat\rho$, given by
\begin{align}\label{E:fidelity-def}
	F(\rho,\hat\rho) = \bigl[\Tr \sqrt{G}\bigr]^2\,,
\end{align}
where for brevity we define $G = \sqrt{\hat\rho}\rho\sqrt{\hat\rho}$. (Note that some authors use the square root of this quantity as the fidelity. Our definition matches Ref.~\cite{Flammia2011}.) The asymptotic cost in sample complexity is far less than the cost of initially obtaining the estimate $\hat\rho$ when $r$ is sufficiently small compared to $d$.

\begin{theorem}
Given a state estimate $\hat\rho$ with $\rank(\hat\rho) = r$, the number of copies $t$ of the state $\rho$ required for estimating $F(\rho,\hat\rho)$ to within $\pm\veps$ with probability $1-\delta$ using single-copy Pauli measurements satisfies 
\begin{align}
	t=O\left(\frac{r^5}{\veps^4}\bigr[d \log(r^2/\delta) + r^2/\delta\bigl]\right) \,.
\end{align}
\end{theorem}
\begin{proof}
The result uses the DFE protocol of Refs.~\cite{Flammia2011, daSilva2011}, modified as mentioned above, where the states $\ket{\phi_j}$ are the eigenstates of $\hat\rho$.

Expand $\hat\rho$ in its eigenbasis as $\hat\rho = \sum_{j=1}^r \lambda_j \ket{\phi_j}\!\bra{\phi_j}$.  Then we have 
\begin{align}
	G = \sum_{j,k=1}^r \sqrt{\lambda_j \lambda_k} \bra{\phi_j} \rho \ket{\phi_k}\ket{\phi_j}\!\bra{\phi_k} \,.
\end{align}

For all $1\leq j\leq k\leq r$, we use direct fidelity estimation to obtain an estimate $\hat{g}_{jk}$ of the matrix element $\bra{\phi_j}\rho\ket{\phi_k}$, up to some additive error $\eps_{jk}$ that is bounded by a constant, $\abs{\eps_{jk}}\le\eps_0$. 

If each estimate is accurate with probability $1-2 \delta/(r^2+r)$, then by the union bound the probability that they are all accurate is at least $1-\delta$. The total number of copies $t$ required for this is 
\begin{align}
	t=O\left(\frac{r^2}{\eps_0^2}\bigr[d \log(r^2/\delta) + r^2/\delta\bigl]\right) \,.
\end{align}

Let $\hat{g}_{jk} = \hat{g}_{kj}^*$, and let 
\begin{align}
	\hat{G} = \sum_{j,k=1}^r \sqrt{\lambda_j \lambda_k} \hat{g}_{jk} \ket{\phi_j}\!\bra{\phi_k}
\end{align}
be our estimate for $G$.  Finally, let $\hat{G}^+ = [\hat{G}]_+$ be the positive part of the Hermitian matrix $\hat{G}$, and let $\hat{F} = \bigl[\Tr\sqrt{\hat{G}^+}\bigr]^2$ be our estimate of the fidelity $F(\rho,\hat\rho)$. Note that we may assume that $\hat F \le 1$, since if it were larger, we can only improve our estimate by just truncating it back to 1. 

We now bound the error of this fidelity estimate.  We can write $\hat{G}$ as a perturbation of $G$, $\hat{G} = G+E$, where 
\begin{align}
	E = \sum_{j,k=1}^r \sqrt{\lambda_j \lambda_k} \eps_{jk} \ket{\phi_j}\!\bra{\phi_k} \,.
\end{align}
First notice that the Frobenius norm of this perturbation is small,
\begin{align}
	\norm{E}_F = \biggl(\sum_{j,k} \lambda_j \lambda_k \abs{\eps_{jk}}^2 \biggr)^{1/2} \le \eps_0 \biggl\vert\sum_j \lambda_j \biggr\rvert = \eps_0 \,.
\end{align}
(If $\hat\rho$ is subnormalized, then this last equality becomes an inequality.)

Next observe that
\begin{align*}
	\abs{F-\hat F} = \Bigl\lvert \Tr(\sqrt{G} + \sqrt{\hat{G}^+}) \Tr(\sqrt{G} - \sqrt{\hat{G}^+}) \Bigr\rvert \le 2 \Delta \,,
\end{align*}
where we define 
\begin{align}
	\Delta = \bigl\lvert\Tr(\sqrt{G} - \sqrt{\hat{G}^+})\bigr\rvert \,.
\end{align}
Using the reverse triangle inequality, we can bound $\Delta$ in terms of the trace norm
\begin{align}
	\Delta \le \bigl\lVert\sqrt{G} - \sqrt{\hat{G}^+}\bigr\rVert_\tr \,.
\end{align}
Using~\cite[Thm. X.1.3]{Bhatia1996} in the first step, we find that
\begin{align*}
	\bigl\lVert\sqrt{G} - \sqrt{\hat{G}^+}\bigr\rVert_\tr \le 
	\Bigl\lVert\sqrt{\abs{G -\hat{G}^+}}\,\Bigr\rVert_\tr \le 
	\sqrt{r}\sqrt{\bigl\lVert G -\hat{G}^+\bigr\rVert_\tr}\,,
\end{align*}
where the second bound follows from the Cauchy-Schwarz inequality on the vector of eigenvalues of $\abs{G -\hat{G}^+}$. Using the Jordan decomposition of a Hermitian matrix into a difference of positive matrices $X = [X]_+ - [X]_-$, we can rewrite $G-\hat{G}^+$ as
\begin{align}
	G-\hat{G}^+ = G-[G+E]_+ = -\eps - [G+E]_- \,.
\end{align}
Then by the triangle inequality and positivity of $G$, 
\begin{align}
	\norm{G-\hat{G}^+}_\tr \le \norm{E}_\tr + \norm{[G+E]_-}_\tr \le 2 \norm{E}_\tr\,.
\end{align}
Using the standard estimate $\norm{E}_\tr \le \sqrt{r}\norm{E}_F$, we find
\begin{align}
	\Delta \le r^{3/4} \sqrt{2 \eps_0} \,.
\end{align}
This gives our desired error bound, albeit in terms of $\eps_0$ instead of the final quantity $\veps$. To get our total error to vanish, we take $\veps = 2 r^{3/4} \sqrt{2 \eps_0}$, which gives us the final scaling in the sample complexity.
\end{proof}

As a final remark, we note that by computing $\Delta$ for the special case $\hat\rho = \rho = \one/r$, $\eps = \eps_0 \one/r$, we find that 
\begin{align}
	\Delta = \sqrt{1+r \eps_0}-1 \,,
\end{align}
so the error bound for this protocol is tight with respect to the scaling in $\eps_0$ (and hence $\veps$). However, we cannot rule out that there are other protocols that achieve a better scaling. Also, it seems that the upper bound for the current protocol with respect to $r$ could potentially be improved by a factor of $r$ with a more careful analysis.

% !TEX root = ../tradeoff.tex

%------------------------------------------------------------------------------------------------------------%
\section{Numerical Simulations\label{S:numerics}}
%------------------------------------------------------------------------------------------------------------%

We numerically simulated the reconstruction of a quantum state given Pauli measurements and the estimators from Eqs.~(\ref{eqn-ds}) and (\ref{eqn-lasso}). Here we compare these estimates to those obtained from a standard maximum likelihood estimate (MLE) as a benchmark. 

As mentioned previously, all of our estimators have the advantage that they are convex functions with convex constraints, so the powerful methods of convex programming~\cite{Boyd2004} guarantee that the exact solution can be found efficiently and certified. We take full advantage of this certifiability and do simulations which can be certified by interior-point algorithms. This way we can separate out the performance of the estimators themselves from the (sometimes heuristic) methods used to compute them on very large scale instances. 

For our estimators, we will explicitly enforce the positivity condition. That is, we will simulate the following slight modifications of Eqs.~(\ref{eqn-ds}) and (\ref{eqn-lasso}) given by
\begin{equation}\label{E:new-ds}
	\hat{\rho}_\ds' = \arg\min_{X\succeq0}\, \Tr X \st \norm{\calA^*(\calA(X)-y)} \leq \lambda\,,
\end{equation}
and
\begin{equation}\label{E:new-lasso}
	\hat{\rho}_\lasso' = \arg\min_{X\succeq0}\, \tfrac{1}{2} \norm{\calA(X)-y}_2^2 + \mu \Tr X \,.
\end{equation}
Moreover, whenever the trace of the resulting estimate is is less than one, we renormalize the state, $\hat\rho \mapsfrom \hat\rho/\Tr(\hat\rho)$. 

%------------------------------------------------------------------------------------------------------------%
\subsection{Setting the Estimator Parameters \texorpdfstring{$\lambda$}{lambda} and \texorpdfstring{$\mu$}{mu}}
%------------------------------------------------------------------------------------------------------------%

From Thm.~\ref{thm-errorbound} we know roughly how we should choose the free parameters $\lambda$ and $\mu$, modulo the unknown constants $C_i$, for which we will have to make an empirical guess. We guess that the log factor is an artifact of the analysis, and that the state is nearly pure. Then we could choose $\lambda \sim \mu \sim d/\sqrt{t}$. For the matrix Dantzig selector of \Eref{E:new-ds}, we specifically  chose $\lambda = 
3 d/\sqrt{t}$ for our simulations. However, this did not lead to very good performance for the matrix Lasso of \Eref{E:new-lasso} when $m$ was large. We chose instead $\mu = 4 m/\sqrt{t}$, which agrees with $\lambda$ when $m \sim d$, as it can be for nearly pure states. 

We stress that very little effort was made to optimize these choices for the parameters. We simply picked a few integer values for the constants in front (namely $2,\dots,5$), and chose the constant that worked best upon a casual inspection of a few data points. We leave open the problem of determining what the optimal choices are for $\lambda$ and $\mu$.

%------------------------------------------------------------------------------------------------------------%
\subsection{Time Needed to Switch Measurement Settings}
%------------------------------------------------------------------------------------------------------------%

Real measurements take time. In a laboratory setting, time is a precious commodity for a possibly non-obvious reason. An underlying assumption in the way we typically formulate tomography is that the unknown state $\rho$ can be prepared identically many times. However, the parameters governing the experiment often drift over a certain timescale, and this means that beyond this characteristic time, it is no longer appropriate to describe the measured ensemble by the symmetric product state $\rho^{\otimes t}$.  

To account for this characteristic timescale, we introduce the following simple model.  We assume that the entire experiment takes place in a fixed amount of time $T$.  In a real experiment, this is the timescale beyond which we cannot confidently prepare the same state $\rho$, or it is the amount of time it takes for the student in the lab to graduate, or for the grant funding to run out, whichever comes first.  Within this time limit, we can either change the measurement \emph{settings}, or we can take more \emph{samples}.  We assume that there is a unit cost to taking a single measurement, but that switching to a different measurement configuration takes an amount of time $c$.

At the one extreme, the switching cost $c$ to change measurement settings could be quite large, so that it is only barely possible to conduct enough independent measurements that tomography is possible within the allotted time $T$.  In this case, we expect that compressed tomography will outperform standard methods, since these methods have not been optimized for the case of incomplete data.  At the other extreme, the switching cost $c$ could be set to $0$ (or some other small value), in which case we are only accounting for the cost of sampling.  Here it is not clear which method is better, for the following reason.  Although the standard methods are able to take a sufficient amount of data in this case, each observable will attain comparatively little precision relative to the fewer observables measured with higher precision in the case of compressed tomography for the same fixed amount of time $T$. One of our goals is to investigate if there are any tradeoffs in this simple model. 

For the relative cost $c$ between switching measurement settings and taking one sample, we use $c = 20$, a number inspired by the current ion trap experiments done by the Blatt group in Innsbruck. However, we note that the conclusions don't seem to be very sensitive to this choice, as long as we don't do something outrageous like $c > t$, which would preclude measuring more than even one observable.

%------------------------------------------------------------------------------------------------------------%
\subsection{Other Simulation Parameters\label{S:params}}
%------------------------------------------------------------------------------------------------------------%

We consider the following ensembles of quantum states. We initially choose $n=5$ qubit states from the invariant Haar measure on $\C^{2^n}$. 
Then we add noise to the state by applying independent and identical depolarizing channels to each of the $n$ qubits. Recall that the depolarizing channel with strength $\gamma$ acts on a single qubit as
\begin{align}
	\mathcal{D}_\gamma(\rho) = \gamma \frac{\one}{2} + (1-\gamma) \rho \,,
\end{align}
that is, with probability $\gamma$ the qubit is replaced by a maximally mixed state and otherwise it is left alone. Our simulations assume very weak decoherence, and we use $\gamma = 0.01$.

We use two measures to quantify the quality of a state reconstruction $\hat\rho$ relative to the underlying true state $\rho$, namely the (squared) fidelity $\lVert \sqrt{\hat\rho}\sqrt{\rho}\rVert^2_\tr$ and the trace distance $\frac{1}{2} \norm{\rho - \hat\rho}_\tr$.

We used the interior-point solver SeDuMi~\cite{Sturm1999} to do the optimizations in Eqs.~(\ref{E:new-ds}) and (\ref{E:new-lasso}). Although these algorithms are not suitable for large numbers of qubits, they produce very accurate solutions, which allows a more reliable comparison between the different estimators.  The data and the code used to generate the data for these plots can be found with the source code for the arXiv version of this paper. For larger numbers of qubits, one may instead use specialized algorithms such as SVT, TFOCS, or FPCA to solve these convex programs~\cite{Cai2010,Becker2010,Ma2011}; however, the quality of the solutions depends somewhat on the algorithm, and it is an open problem to explore this in more detail.

For the MLE, we used the iterative algorithm of Ref.~\cite{Rehacek2001}, which has previously been used in e.g.\ Ref.~\cite{Molina-Terriza2004}. This method converged on every example we tried, so we did not have to use the more sophisticated ``diluted'' method of Ref.~\cite{Rehacek2007} that guarantees convergence. 

For the purposes of our simulations, we sampled our Pauli operators \emph{without} replacement.

%------------------------------------------------------------------------------------------------------------%
\subsection{Results and Analysis}
%------------------------------------------------------------------------------------------------------------%

\begin{figure}[t!]
\centering
\begin{tabular}{c}
a)\hspace{-10pt} \includegraphics[scale=.42]{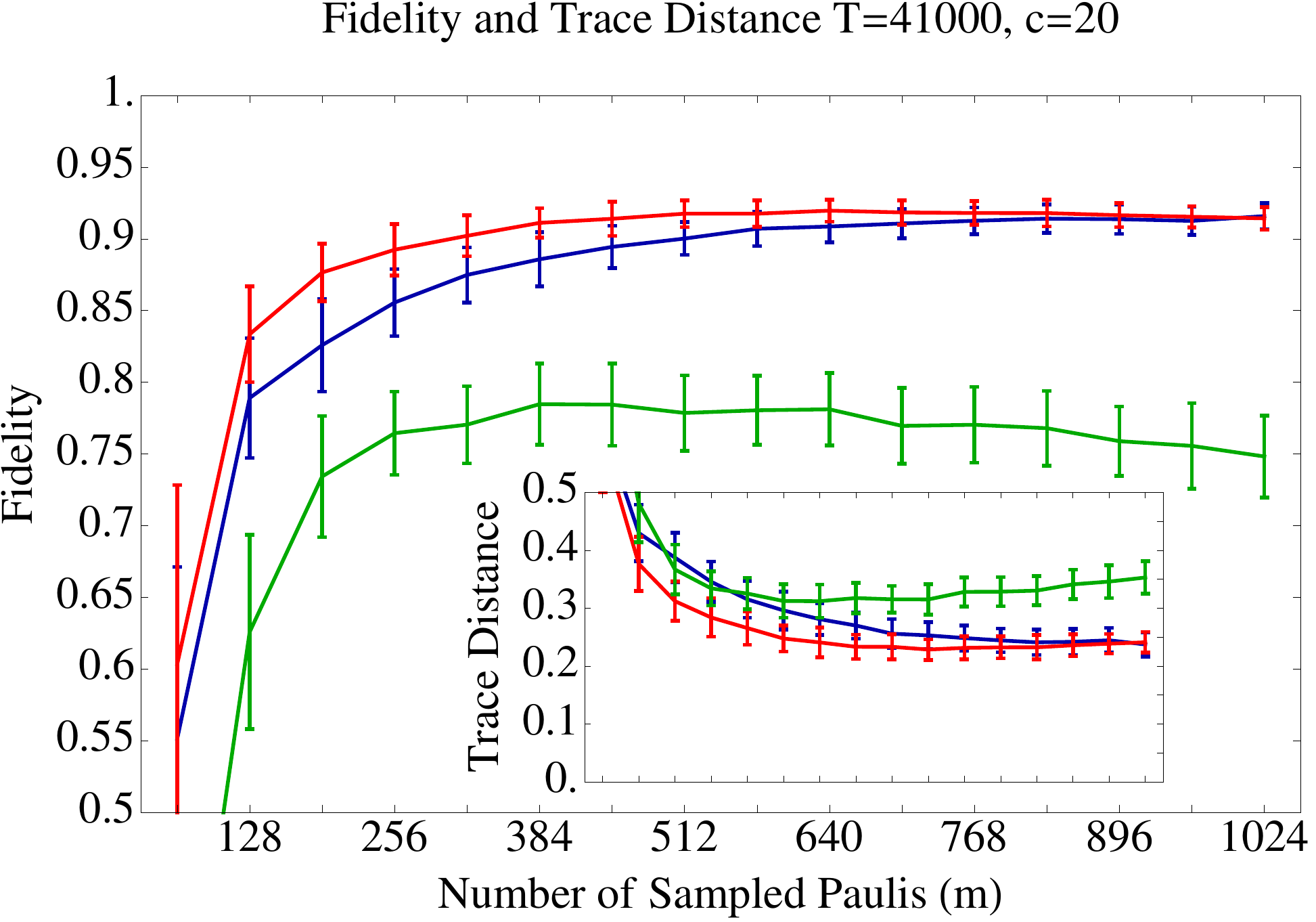} \vspace{5pt}\\
b)\hspace{-10pt} \includegraphics[scale=.42]{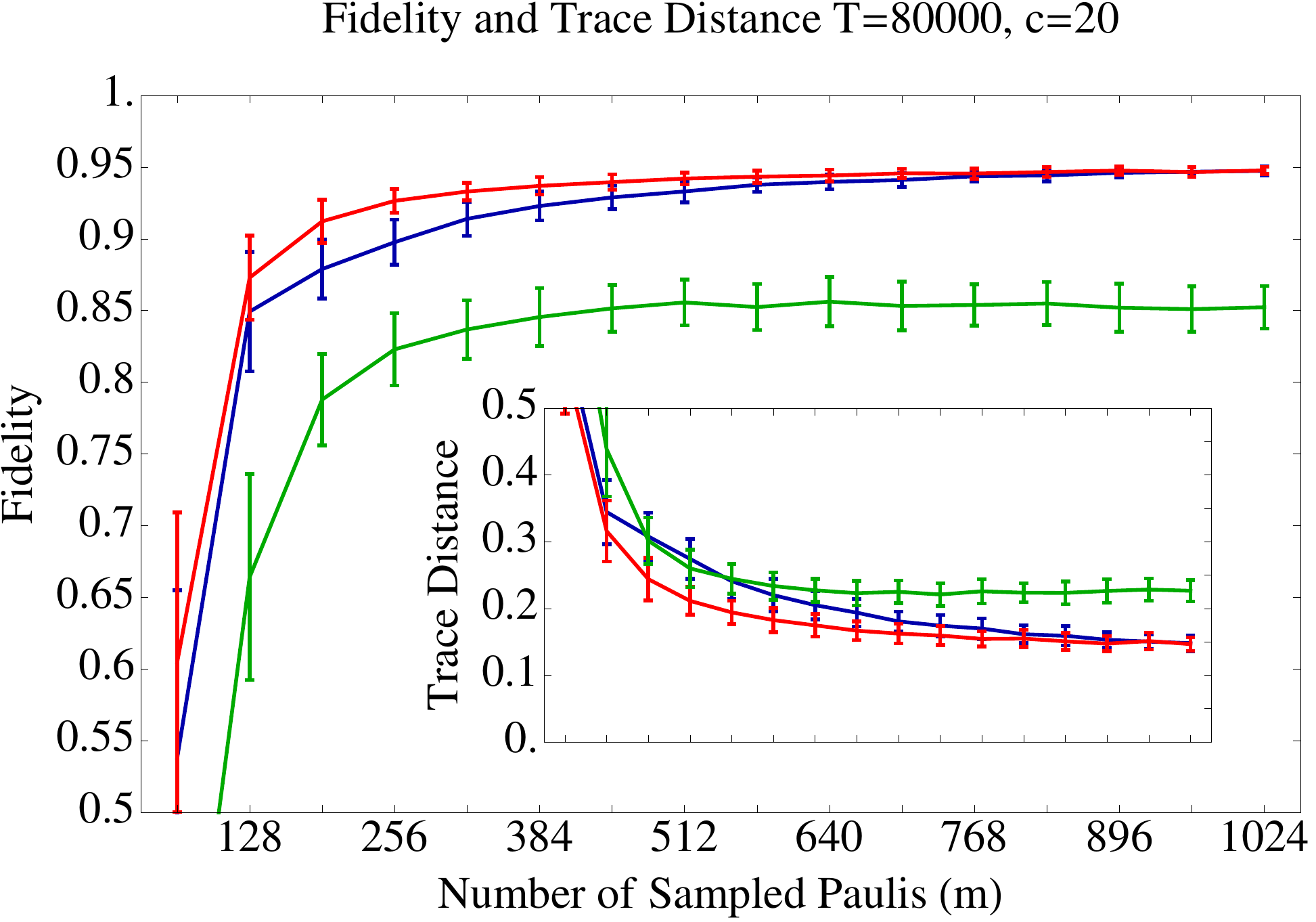}  \vspace{5pt}\\
c)\hspace{-10pt} \includegraphics[scale=.42]{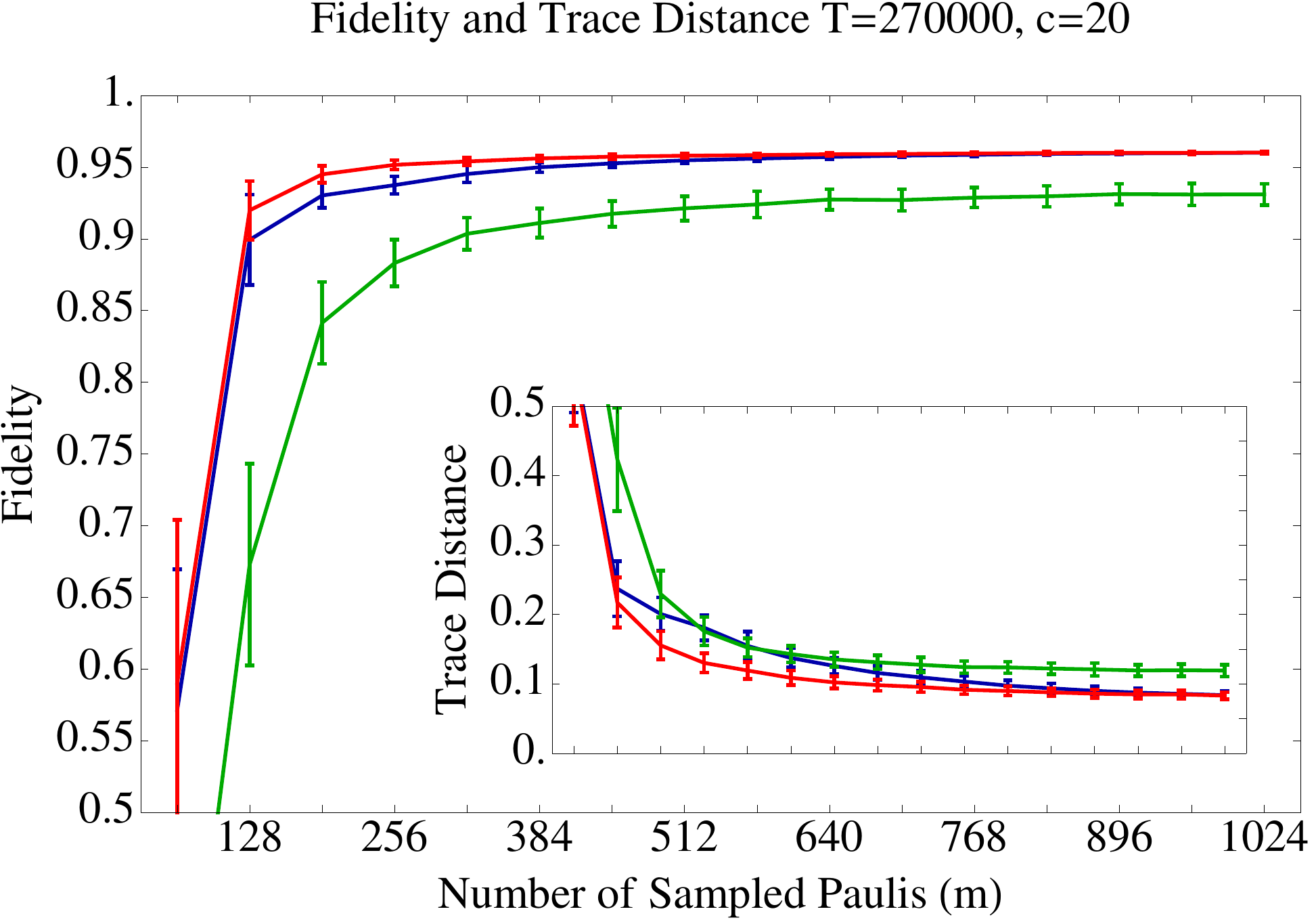}
\end{tabular}
\begin{caption}{\label{F:numerics}Fidelity and trace distance (inset) as a function of $m$, the number of sampled Paulis. Plots a), b) and c) are for a total measurement time of $T=41$k, $80$k, and $270$k respectively, in units where the cost to measure a single copy is one unit of time, and the cost to switch measurement settings is $c=20$ units. The three estimators plotted are the matrix Dantzig selector (\Eref{E:new-ds}, blue), the matrix Lasso (\Eref{E:new-lasso}, red), and a standard MLE (green). Each bar shows the mean and standard deviation from 120 Haar-random 5-qubit states with 1\% local depolarizing noise. Our estimators consistently outperform MLE, even after optimizing the MLE over $m$. See the main text for more details.\vspace{-5pt}}
\end{caption}
\end{figure}

The data are depicted in Fig.~\ref{F:numerics}. The three subfigures a--c use three different values for the total sampling time $T$ in increasing order. We have plotted both the fidelity and the trace distance (inset) between the recovered solution and the true state. %We note that a simple least squares estimator performed worse than all the others, so we do not include a plot of that here.

Several features are immediately apparent. First, we see that both of the compressed sensing estimators consistently outperform MLE along a wide range of values of $m$, the number of Paulis sampled, even when we choose the optimal value of $m$ just for the MLE. Once $m$ is sufficiently large (but still $m \ll d^2$) \emph{all} of the estimators converge to a reasonably high fidelity estimate. Thus, it is not just the compressed sensing estimators which are capable of reconstructing nearly low rank states with few measurement settings, but rather this seems to be a generic feature of many estimators. However, the compressed sensing estimators seem to be particularly well-suited for the task. 

When the total amount of time $T$ is \emph{very} small (Fig.~\ref{F:numerics}a), then there is a large advantage to using compressed sensing. In this regime, there is barely enough time to make one measurement per setting if we insist on making an informationally complete measurement, so the measurement statistics aren't even Gaussian. Thus, the fluctuations are so large that trying to fit a density operator to these data just leads to poor results because you wind up fitting only noise. While the advantage of compressed sensing is clear in this regime, it is not applicable when we are interested in extremely high-fidelity state reconstructions (namely greater than 95\%). 

As the total amount of time available increases, all of the estimators of course converge to higher fidelity estimates. Interestingly, for the compressed sensing estimators, there seems to be \emph{no tradeoff whatsoever} between the number of measurement settings and the fidelity once $T$ and $m$ are sufficiently large. That is, the curve is completely flat as a function of $m$ above a certain cutoff. This is most notable for the matrix Lasso, which consistently performs at least as well as the matrix Dantzig selector, and often better. These observations are consistent with the bounds proven earlier. 

The flat curve as a function of $m$ is especially interesting, because it suggests that there is no real drawback to using small values of $m$. Smaller values of $m$ are attractive from the computational point of view because the runtime of each reconstruction algorithm scales with $m$. This makes a strong case for adopting these estimators in the future, and at the very least more numerical studies are needed to investigate how well these estimators perform more generally.

We draw the following conclusion from these simulations. The best performance for a fixed value of $T$ is given by the matrix Lasso estimator of Eqs.~(\ref{eqn-lasso}) and (\ref{E:new-lasso}) in the regime where $m$ is nearly as small as possible. Here the fidelity is larger than the other estimators (if only by a small or negligible amount when $T$ is large), but using a small value for $m$ means smaller memory and processing time requirements when doing the state reconstruction. MLE consistently underperforms the compressed sensing estimators, but still seems to ``see'' the low-rank nature of the underlying state and converges to a reasonable estimate even when $m$ is small. 

% !TEX root = ../tradeoff.tex

%------------------------------------------------------------------------------------------------------------%
\section{Process Tomography}\label{S:process}
%------------------------------------------------------------------------------------------------------------%

Compressed sensing techniques can also be applied to quantum \textit{process} tomography.  Here, our method has an advantage when the unknown quantum process has small Kraus rank, meaning it can be expressed using only a few Kraus operators.  This occurs, for example, when the unknown process consists of unitary evolution (acting globally on the entire system) combined with local noise (acting on each qubit individually, or more generally, acting on small subsets of the qubits).

Consider a system of $n$ qubits, with Hilbert space dimension $d = 2^n$.  Let $\calE$ be a completely positive (CP) map from $\CC^{d\times d}$ to $\CC^{d\times d}$, and suppose that $\calE$ has Kraus rank $r$.  Using compressed sensing, we will show how to characterize $\calE$ using $m = O(r d^2 \log d)$ settings.  (For comparison, standard process tomography requires $d^4$ settings, since $\calE$ contains $d^4$ independent parameters in general.)  Furthermore, our compressed sensing method requires only the ability to prepare product eigenstates of Pauli operators and perform Pauli measurements, and it does not require any ancilla qubits.

We remark that, except for the ancilla-assisted method, just the notion of ``measurement settings'' for process tomography does not capture all of the complexity because of the need to have an \emph{input} to the channel. Here we define one ``input setting'' to be a basis of states from which the channel input should be sampled uniformly. Then the total number of settings $m$ is the sum of the number of measurement settings (Paulis, in our case) and input settings. This definition justifies the claim that the number of settings for our compressed process tomography scheme is $m = O(rd^2 \log d)$. 

The analysis here focuses entirely on the number of settings $m$. We forgo a detailed analysis of $t$, the sample complexity, and instead leave this open for future work.

Note that there is a related set of techniques for estimating an unknown process that is \textit{elementwise sparse} with respect to some known, experimentally accessible basis~\cite{Shabani2011}.  These techniques are not directly comparable to ours, since they assume a greater amount of prior knowledge about the process, and they use measurements that depend on this prior knowledge.  We will discuss this in more detail at the conclusion of this section.

%------------------------------------------------------------------------------------------------------------%
\subsection{Our Method}
%------------------------------------------------------------------------------------------------------------%

First, recall the Jamio\l{}kowski isomorphism~\cite{Jamiolkowski1972}:  the process $\calE$ is completely and uniquely characterized by the state 
\[
	\rho_\calE = (\calE\ox\calI)(\ket{\psi_0}\bra{\psi_0}), 
\]
where $\ket{\psi_0} = \frac{1}{\sqrt{d}} \sum_{j=1}^{d} \ket{j}\ox\ket{j}$.  Note that when $\calE$ has Kraus rank $r$, the state $\rho_\calE$ has rank $r$.  This immediately gives us a way to do compressed quantum process tomography: first prepare the Jamio\l{}kowski state $\rho_\calE$ (by adjoining an ancilla, preparing the maximally entangled state $\ket{\psi_0}$, and applying $\calE$); then do compressed quantum state tomography on $\rho_\calE$; see Figure \ref{fig-qptancilla}.

\begin{figure}
\ifx\JPicScale\undefined\def\JPicScale{1}\fi
\unitlength \JPicScale mm
\begin{picture}(40,27.5)(0,0)
\put(0,5){\makebox(0,0)[cc]{}}

\put(0,5){\makebox(0,0)[cc]{}}

\put(0,15){\makebox(0,0)[cc]{$\ket{\psi_0}$}}

\linethickness{0.3mm}
\multiput(5,15)(0.12,-0.21){42}{\line(0,-1){0.21}}
\linethickness{0.3mm}
\put(10,6.25){\line(1,0){22.5}}
\linethickness{0.3mm}
\put(32.5,2.5){\line(1,0){7.5}}
\put(32.5,2.5){\line(0,1){7.5}}
\put(40,2.5){\line(0,1){7.5}}
\put(32.5,10){\line(1,0){7.5}}
\put(36.25,6.25){\makebox(0,0)[cc]{$P_B$}}

\linethickness{0.3mm}
\multiput(5,15)(0.12,0.21){42}{\line(0,1){0.21}}
\linethickness{0.3mm}
\put(10,23.75){\line(1,0){5}}
\linethickness{0.3mm}
\put(22.5,23.75){\line(1,0){10}}
\linethickness{0.3mm}
\put(15,27.5){\line(1,0){7.5}}
\put(15,20){\line(0,1){7.5}}
\put(22.5,20){\line(0,1){7.5}}
\put(15,20){\line(1,0){7.5}}
\linethickness{0.3mm}
\put(32.5,20){\line(1,0){7.5}}
\put(32.5,20){\line(0,1){7.5}}
\put(40,20){\line(0,1){7.5}}
\put(32.5,27.5){\line(1,0){7.5}}
\put(18.75,23.75){\makebox(0,0)[cc]{$\calE$}}

\put(36.25,23.75){\makebox(0,0)[cc]{$P_A$}}

\put(20,23.75){\makebox(0,0)[cc]{}}

\end{picture}
\caption{\label{fig-qptancilla}Compressed quantum process tomography using an ancilla.  The quantum circuit represents a single measurement setting, where one measures the observable $P_A$ on the system and $P_B$ on the ancilla.}
\end{figure}
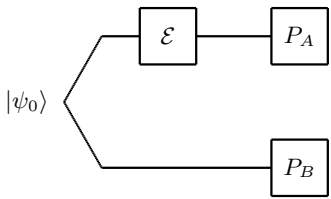

We now show a more direct implementation of compressed quantum process tomography that is equivalent to the above procedure, but does not require an ancilla. Observe that in the above procedure, we need to estimate expectation values of the form 
\begin{equation*}
\Tr\bigl((P_A\ox P_B) \rho_\calE\bigr) = \Tr\bigl((P_A\ox P_B) (\calE\ox\calI) (\ket{\psi_0}\bra{\psi_0})\bigr), 
\end{equation*}
where $P_A$ and $P_B$ are Pauli matrices.  By using the Kraus decomposition, it is straightforward to derive the equivalent expression
\begin{align}\label{eqn-qpt-ev}
	\Tr((P_A \ox P_B) \rho_\calE) = \frac{1}{d} \Tr(P_A \calE(\overline{P_B})), 
\end{align}
where the bar denotes complex conjugation in the standard basis.

%\begin{equation*}
%\begin{split}
%\Tr((P_A&\tnsr P_B) \rho_\calE) \\
% &= \sum_\ell \Tr[((P_A K_\ell) \tnsr P_B) \ket{\psi_0}\bra{\psi_0} (K_\ell^\dagger \tnsr I)] \\
% &= \sum_\ell \bra{\psi_0} ((K_\ell^\dagger P_A K_\ell) \tnsr P_B) \ket{\psi_0} \\
% &= \sum_\ell \frac{1}{d} \sum_{j,j'=0}^{d-1} \bra{j}P_B\ket{j'} \bra{j} K_\ell^\dagger P_A K_\ell \ket{j'} \\
% &= \sum_\ell \frac{1}{d} \sum_{j,j'=0}^{d-1} \bra{j'}\overline{P_B}\ket{j} \bra{j} K_\ell^\dagger P_A K_\ell \ket{j'} \\
% &= \sum_\ell \frac{1}{d} \Tr(\overline{P_B} K_\ell^\dagger P_A K_\ell) \\
% &= \sum_\ell \frac{1}{d} \Tr(P_A K_\ell \overline{P_B} K_\ell^\dagger)
% = \frac{1}{d} \Tr(P_A \calE(\overline{P_B})), 
%\end{split}
%\end{equation*}
%where we used the Kraus-operator representation $\calE(\rho) = \sum_\ell K_\ell \rho K_\ell^\dagger$, and the fact that $\bra{j}P_B\ket{j'} = \overline{\bra{j'}P_B^\dagger\ket{j}} = \bra{j'}\overline{P_B}\ket{j}$ (since $P_B$ is Hermitian, and $\ket{j}$, $\ket{j'}$ are real).  

We now show how to estimate the expectation value (\ref{eqn-qpt-ev}).  Let $\lambda_j$ and $\ket{\phi_j}$ denote the eigenvalues and eigenvectors of $\overline{P_B}$.  Then we have 
\begin{equation*}
\Tr((P_A\ox P_B) \rho_\calE)
 = \frac{1}{d} \sum_{j=1}^{d} \lambda_j \Tr(P_A \calE(\ket{\phi_j}\bra{\phi_j})).
\end{equation*}
To estimate this quantity, we repeat the following experiment many times, and average the results:  choose $j\in [d]$ uniformly at random, prepare the state $\ket{\phi_j}$, apply the process $\calE$, measure the observable $P_A$, and multiply the measurement result by $\lambda_j$.  (See Figure \ref{fig-qptdirect}.)  In this way, we learn the expectation values of the Jamio\l{}kowski state $\rho_\calE$ without using an ancilla.  We then use compressed quantum state tomography to learn $\rho_\calE$, and from this we recover $\calE$.

\begin{figure}
\ifx\JPicScale\undefined\def\JPicScale{1}\fi
\unitlength \JPicScale mm
\begin{picture}(41.25,10.75)(0,0)

\put(0,7){\makebox(0,0)[cc]{$\ket{\phi_j(\overline{P_B})}$}}

\linethickness{0.3mm}
\put(7.5,7){\line(1,0){8.75}}
\linethickness{0.3mm}
\put(23.75,7){\line(1,0){10}}
\linethickness{0.3mm}
\put(16.25,10.75){\line(1,0){7.5}}
\put(16.25,3.25){\line(0,1){7.5}}
\put(23.75,3.25){\line(0,1){7.5}}
\put(16.25,3.25){\line(1,0){7.5}}
\linethickness{0.3mm}
\put(33.75,3.25){\line(1,0){7.5}}
\put(33.75,3.25){\line(0,1){7.5}}
\put(41.25,3.25){\line(0,1){7.5}}
\put(33.75,10.75){\line(1,0){7.5}}
\put(20,7){\makebox(0,0)[cc]{$\calE$}}

\put(37.5,7){\makebox(0,0)[cc]{$P_A$}}

\end{picture}
\caption{\label{fig-qptdirect}Compressed quantum process tomography, implemented directly without an ancilla.  Here one prepares a random eigenstate of $\overline{P_B}$, applies the process $\calE$, and measures the observable $P_A$ on the output.}
\end{figure}
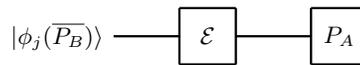

%------------------------------------------------------------------------------------------------------------%
\subsection{Related Work\label{S:process-related-work}}
%------------------------------------------------------------------------------------------------------------%

Our method is somewhat different from the method described in Ref.~\cite{Shabani2011}. Essentially the difference is that our method works for any quantum process with small Kraus rank, whereas the method of Shabani \textit{et al}.\ works for a quantum process that is elementwise sparse in a known basis (provided this basis is experimentally accessible in a certain sense). The main advantage of the Shabani \textit{et al}.\ method is that it can be much faster: for a quantum process $\calE$ that is $s$-sparse (i.e., has $s$ nonzero matrix elements), it requires only $O(s\log d)$ settings. The main disadvantage is that it requires more prior knowledge about $\calE$, and is more difficult to apply.  While it has been demonstrated in a number of scenarios, there does not seem to be a general recipe for designing measurements that are both experimentally feasible and effective in the sparse basis of $\calE$.

To clarify these issues, we now briefly review the Shabani \textit{et al}.\ method.  We assume that we know a basis $\Gamma = \set{\Gamma_\alpha \;|\; \alpha \in [d^2]}$ in which the process $\calE$ is $s$-sparse.  For example, when $\calE$ is close to some unitary evolution $U$, one can construct $\Gamma$ using the SVD basis of $U$.  This guarantees that, if $\calE$ contains no noise, it will be perfectly sparse in the basis $\Gamma$. However, in practice, $\calE$ will contain noise, which need not be sparse in the basis $\Gamma$; any non-sparse components will not in general be estimated accurately.  The success of the Shabani \textit{et al}.\ method therefore rests on the assumption that the \textit{noise} is also sparse in the basis $\Gamma$. Although this assumption has been verified in a few specific scenarios, it seems less clear why it should hold in general. By comparison, our method simply assumes that the noise is described by a process matrix that is low rank; this can be rigorously justified for any noise process that involves only local interactions or few-body processes.

The other complication with the Shabani \textit{et al}.\ method concerns the design of the state preparations and measurements. On one hand, these must satisfy the RIP condition for $s$-sparse matrices over the basis $\Gamma$; on the other hand, they must be easy to implement experimentally. This has been demonstrated in some cases, by using states and measurements that are ``random'' enough to show concentration of measure behaviors, but also have tensor product structure. However, these constructions are not guaranteed to work in general for an arbitrary basis $\Gamma$.

We leave open the problem of doing a comparative study between these and other methods~\cite{Mohseni2006,Mohseni2007}, akin to Ref.~\cite{Mohseni2008}.

% !TEX root = ../tradeoff.tex

%------------------------------------------------------------------------------------------------------------%
\section{Conclusion}\label{S:conclusion}
%------------------------------------------------------------------------------------------------------------%

In this work, we have introduced two new estimators for compressed tomography: the matrix Dantzig selector and the matrix Lasso (Eqs.~(\ref{eqn-ds}) and (\ref{eqn-lasso}).) We have proved that the sample complexity for obtaining an estimate that is accurate to within $\eps$ in the trace distance scales like $O(\tfrac{r^2 d^2}{\eps^2} \log d)$ for rank-$r$ states, and that for higher rank states, the additional error is proportional to the truncation error. This error scaling is optimal up to constant multiplicative factors, and requires measuring only $O(r d \poly\log d)$ Pauli expectation values, a fact we proved using the RIP~\cite{Liu2011}. We also proved that our sample complexity upper bound is within $\poly \log d$ of the sample complexity of the optimal minimax estimator, where the risk function is given by a trace norm confidence interval. We showed how a modification of direct fidelity estimation can be used to unconditionally certify the estimate using a number of measurements which is asymptotically negligible compared to obtaining the original estimate. We numerically simulated our estimators and found that they outperform MLE, giving higher fidelity estimates from the same amount of data. Finally, we generalized our method to quantum process tomography using only Pauli measurements and preparation of product eigenstates of Pauli operators. 

There are many interesting open questions that remain. On the theoretical side, one open problem is of course to tighten the gap that remains between the sample complexity upper and lower bounds. Another open problem is to try to prove optimality with respect to alternative criteria other than minimax risk. For example, it would be interesting to find a useful notion of average case optimality. 

One major open problem is to switch focus from two-outcome Pauli measurements to alternative measurements which are still experimentally feasible. For example, measurements in a local basis have $2^n$ outcomes and are not directly analyzable using our techniques. It would be very interesting to give an analysis of our estimators from the perspective of such local basis measurements. One difficulty, however, is that something like the RIP is not likely to hold in this case, so we will need additional techniques.

On the numerical side, some of the open questions are the following. First, it would be very interesting to do a more detailed numerical study of the performance of our estimators. While they have clearly outperformed MLE in the simulations reported here, there is no question that this is a narrow range of parameters on which we have tested these estimators. It would be interesting to do additional comparative studies between these and other estimators to see how robust these performance enhancements are. It would also be very interesting to study fast first-order solvers such as Refs.~\cite{Cai2010,Becker2010,Ma2011} which could compute estimates on a large number of qubits (10 or more). 

The success of the estimators we studied depends on being able to find good values for the parameters $\lambda$ and $\mu$. While we have used simple heuristics to pick particular values, a detailed study of the optimal values for these parameters could only improve the quality of our estimators. Moreover, MLE seems to enjoy the same ``plateau'' phenomenon, where the quality of the estimate is insensitive to $m$ above a certain cutoff. This leads us to speculate that this is a generic phenomenon among many estimators, and that perhaps there are even better choices for estimators than the ones we benchmark here.

\acknowledgements

We thank R.~Blume-Kohout, M.~Kleinmann and T.~Monz for discussions and B.~Brown for some preliminary numerical investigations. We additionally thank B.~Englert for hosting the workshop on quantum tomography in Singapore where some of this work was completed. STF was supported by the IARPA MQCO program. DG was supported by the Excellence Initiative of the German Federal and State Governments (grant ZUK~43). Contributions to this work by NIST, an agency of the US government, are not subject to copyright laws. JE was supported by EURYI, the EU (Q-Essence), and BMBF (QuOReP).

%------------------------------------------------------------------------------------------------------------%
\bibliographystyle{titles}
\bibliography{tradeoff}
%------------------------------------------------------------------------------------------------------------%

\end{document}